\documentclass[conference]{IEEEtran}                                             
\usepackage{amsmath}
\usepackage{graphics}
\usepackage{graphicx}
\usepackage{amssymb}
\usepackage{amsthm}
\usepackage[usenames,dvipsnames,svgnames,table]{xcolor}
\usepackage{tikz}
\usepackage{multirow}
\usepackage{subcaption}
\usepackage{cite}

\newcommand \edit[1] {{{{#1}}}}

\newcommand{\LyX}{L\kern-.1667em\lower.25em\hbox{Y}\kern-.125emX\spacefactor1000}

\title{The Interplay of Competition and Cooperation Among Service Providers}

\author{
  \IEEEauthorblockN{
    Mohammad Hassan Lotfi\IEEEauthorrefmark{1}, 
    Xingran Chen\IEEEauthorrefmark{2}, 
   Saswati Sarkar \IEEEauthorrefmark{1}
  \IEEEauthorblockA{
    \IEEEauthorrefmark{1}Department of Electrical and System Engineering, University of Pennsylvania, Philadelphia, PA, 19104 \\
      Emails: lotfm@seas.upenn.edu \& swati@seas.upenn.edu}
    \IEEEauthorblockA{\IEEEauthorrefmark{2} Applied Mathematics and Computational Science, University of Pennsylvania, Philadelphia, PA, 19104 \\
Email: xingranc@sas.upenn.edu
     }
 }
 \thanks{Parts of this work was presented in CISS'17 \cite{lotfi2017economics}.}}

\IEEEoverridecommandlockouts
 
\begin{document}
\maketitle
\newtheorem{lemma}{Lemma}
\newtheorem{note}{Note}
\newtheorem{property}{Property}
\newtheorem{theorem}{Theorem}
\newtheorem{definition}{Definition}
\newtheorem{corollary}{Corollary}
\newtheorem{remark}{Remark}
\newtheorem{assumption}{Assumption}

\begin{abstract}
We consider the economics of the interaction between Mobile Virtual Network Operators (MVNOs) and Mobile Network Operators (MNOs). We investigate the incentives of an MNO for offering some of her resources to an MVNO instead of using the resources for her own End-Users (EUs). We consider a market with one MNO and one MVNO, and a continuum of undecided EUs. Two cases for EUs are considered:  (i) when EUs need to choose either the MNO or the MVNO, and (ii) when EUs have an outside option. In each of these cases, we consider a non-cooperative framework of sequential game and a cooperative framework of  bargaining game. We characterize the Subgame Perfect Nash Equilibria (SPNE) and Nash Bargaining Solution (NBS) of the sequential and bargaining games, respectively. We show that in the non-cooperative framework, SPNE assumes two forms: (i) the MNO invests minimally on her infrastructure and the MVNO leases all the newly invested resources and (ii) the MNO invests more on her infrastructure, the MVNO leases only part of these resources to the MVNO, and the MNO uses the rest herself to attract EUs. Thus, in  both, the MNO generates revenue indirectly through the MVNO. In addition to that, in (ii),  the MNO generates revenue directly from the EUs. We also prove that in the bargaining framework,  the MVNO either reserves all the resources or no resources from the MNO, and the MNO's investments are guided by whether  EUs have an outside option. If they don't,  then the MNO invests as little as possible on her infrastructure. If they do, then the MNO invests more.
\end{abstract}

\section{Introduction}

Traditionally, wireless services have been offered by Service Providers (SPs) that own the infrastructure they are operating on.  Nowadays SPs are divided into (i) Mobile Network Operators (MNOs) that own the infrastructure, and (ii) 
Mobile Virtual Network Operators (MVNOs)  that do not own the infrastructure they are operating on, and use the resources of one or more MNOs based on a business contract. MVNOs can distinguish their plans from MNOs by  bundling their service with other products,  offering different pricing plans  for End-Users (EUs), or building  a good reputation through a better customer service. In recent years, the number of MVNOs has been rapidly growing. According to \cite{GSMA}, between June 2010 and June 2015, the number of MVNOs increased by 70 percent worldwide,  reaching 1,017 as of June 2015.  Even some MNOs developed their own MVNOs. An example of which is Cricket wireless which is owned by At\&t and offers a prepaid wireless service to EUs. Another example of MVNOs is the Google's Project Fi in which the customer's service is handled using Wifi hotspots wherever/whenever they exist. If no Wifi is available, then the service is handled using the resources of Sprint, T-Mobile or U.S. Cellular networks, whichever has the better service at that particular location/time.

In this work, we consider the economics of the interaction between  MVNOs and MNOs.  We investigate the incentives of an MNO for offering some of her resources to an MVNO instead of using the resources for her own. Thus, we consider the interplay of competition and cooperation between an MNO and an MVNO. The goal is to investigate that under what parameters the competition between SPs overshadow their cooperation and vice versa.  We investigate this interplay between competition and cooperation through cooperative and non-cooperative game frameworks. Note that it is not apriori known  how much an MNO is willing to invest on the infrastructure and how much an MVNO is willing to lease. More importantly, it is not apriori clear that under what conditions the MNO prefers to generate her revenue through the MVNO by leasing the resources to her, and therefore letting the MVNO to attract EUs.

Many works have considered the economics of resource/spectrum sharing and the subsequent profit sharing between SPs. Examples are \cite{cano2016cooperative,berry2013nature,le2009pricing,singh2012cooperative,korcak2012collusion
,banerjee2009voluntary,lotfi2016uncertain,lotfi2017economics_net,yan2011sequential,
yan2013dynamic,tao2013spectrum}. In these works, authors model the environment using game theory and seek to provide intuitions for the pricing schemes, the split of EUs/benefits between SPs, and the investment level of different SPs. In this work, however, as mentioned before, we focus on the competition between the MNOs and MVNOs and its interplay with  cooperation among them.

We consider a market with one MNO (e.g. AT\&T) and  one MVNO (e.g. Google's Project Fi). The MNO decides on investing new resources, and the MVNO leases some of these  newly invested resources   in exchange of a fee.   Both MNO and MNVO  simultaneously decide on their pricing strategies for EUs, and EUs choose one of the MNO or MVNO to buy the wireless plan from. First, we assume the per resource fee that the MVNO pays to the MNO to be fixed. Then, we discuss about the implications of this fee and frameworks for characterizing it.

To model EUs, we consider the hotelling model, in which we consider a continuum of undecided EUs that decide which of the SPs they want to buy their wireless plan from. We assume that EUs have different preferences for each SP. These preferences can be because of different services that SPs offer. For example, the MVNO can bundle the wireless service with a free or cheap international call plan (through VoIP) to make her service more favorable for some EUs. Moreover, the preferences can be because of the reluctance of  an EU to buy her wireless plan from a particular SP (e.g. an  SP with an infamous customer service). We assume that the preferences also depend on the investment level of the SPs. In other words, the higher the investment level of an SP, the lower would be the reluctance of an EU for that SP. 

\edit{We consider two cases. In the first case, SPs compete to attract EUs modeled using the hotelling model, and the EUs have to choose one the SPs (i.e., there does not exist an outside option). In the second case, we consider the demand of SPs to be generated both through the hotelling model and also a demand function. The effects of the demand function are two-fold. First, the demand function models the attrition in the number of EUs of SPs if the investment or price of both SPs are not desirable for EUs. Thus, in effect, an EU may opt for neither SP if neither offers a price-quality combo that is to his satisfaction, which is equivalent to the existence  of an outside option for them. Second, the demand function models an exclusive additional customer base for each of the SPs to draw from depending on her investment and the price she offers. In each of the two cases, concerning an outside option, we consider both non-cooperative and cooperative frameworks of interaction. In the non-cooperative scenario, we consider a sequential game in which the SPs compete over the EUs and decide on the level of cooperation individually and sequentially. In the cooperative scenario, the SPs still compete over the EUs, but they jointly decide their investments and how to share the total proceeds  through a bargaining game which determines the shares in accordance with the investment levels and the EU subscription for each.
}

\edit{In this setting, different outcomes may occurs. 
For example, if EUs in the hotelling model have  high preferences for the MVNO, e.g. because of a good customer service, then the MNO may prefer to lease some of her resources to the MVNO and receives her share of profit through the MVNO, instead of competing for EUs by lowering her price. On the other hand, if EUs have high preferences for the MNO or she can attract more exclusive EUs by using the resources herself, then no cooperation between the MNO and the MVNO occurs. }

We start with the first case in which EUs are modeled using the hotelling model. 
We formulate the game as a sequential game (Section~\ref{section:model}), and seek the Sub-game Perfect Nash Eequilibrium (SPNE) of the sequential
game using backward induction (Section~\ref{section:analysis}).  We prove that the SPNE outcome of the game exists, is unique, and is of the form of two possible outcomes which we characterize. In one of the outcomes, the MNO invests minimally on her infrastructure and  the MVNO leases all the newly invested resources from the MNO. Thus, in this case, the MNO prefers to generate revenue through the MVNO instead of directly generating revenue from EUs. On the other hand, in the second outcome, MNO invests more on her infrastructure, leases parts of these resources to the MVNO and uses the rest to attract EUs herself. Therefore, in this case, MNO generates parts of her revenue through the MVNO and the rest directly through the EUs. Numerical results in Section~\ref{section:numerical} reveal that the former outcome occurs when the per resource fee that the MVNO pays to the MNO is smaller than a threshold, and the latter outcome occurs when the per resource fee is higher than the threshold. Thus, the per resource fee is one of the most important  parameters that influences the interplay of competition and cooperation between MNOs and MVNOs. 


In addition  to the aforementioned non-cooperative framework, we consider a bargaining  framework to determine the amount of cooperation between SPs and the split of profit between them (Section~\ref{section:bargaining}). We characterize the Nash Bargaining Solution (NBS) of the game, and prove that when EUs do not have an outside option, the bargaining framework yields a collusive outcome in which the MNO invests as little as possible (zero, without external or regulatory forces). In addition, the MVNO either reserves all the resources or no resources from the MNO. 

In Section \ref{section:generalization}, we focus on the second case in which the demand of SPs are generated both through the hotelling model and also a demand function. In the non-cooperative framework, we characterize the SPNE outcome of the game. Results reveal that the general behavior of the SPNE outcome would be similar to that of the previous model, i.e. EUs with no outside option. We also analyze the bargaining framework associated with this case. Results of the bargaining framework in this case are slightly different from the previous case. In this case, the results of the NBS yields  a milder version of the aforementioned collusive outcome  in which the MNO no longer opt for the lowest investment level. However,  the MVNO still either reserves all the resources or no resources from the MNO. We conclude the paper in Section~\ref{section:conclusion}.

\section{Model}\label{section:model}
We now present a system model that we use in  Section~\ref{section:analysis}. This model is also updated and used in Section~\ref{section:generalization}. First, we model the payoffs and strategies of SPs, then we model the EUs and their decision making process. Finally, at the end of the section, we formulate the interaction between different entities of the game using a non-cooperative sequential game framework. 
\subsection*{SPs:} 
We consider one Mobile Network Operator (MNO) and one Mobile Virtual Network Operator (MVNO) that compete in attracting a pool of undecided EUs. We denote the MNO by SP$_{L}$ ($L$ represents Leader, since this SP is the leader of the game), and denote the MVNO by SP$_{F}$ ($F$ represents Follower, since this SP is the follower of this leader/follower game). SP$_L$ owns the infrastructure, invests in her infrastructure to attract EUs, and can lease parts of the new resources to SP$_F$ in exchange of money. 

We denote the fee per resources that SP$_{F}$ pays to SP$_{L}$ by $s$, the fraction of EUs that SP$_{F}$ and SP$_L$  attract by $n_F$ and $n_L$, respectively, and the access fees that SP$_F$ and SP$_L$ charge the EUs by $p_F$ and $p_L$, respectively.  The utility of SPs is increasing with the revenue from EUs, which depends on the number of EUs and the access fee. The utility of SP$_F$ (respectively, SP$_L$)  is decreasing (respectively, increasing) with respect to the fee that SP$_F$ pays to SP$_L$ to reserve resources. 
In addition, the utility of SP$_L$ is decreasing with the cost of investing on the infrastructure. Naturally, we expect the cost of reserving resources and investing in the infrastructure to be strictly convex, i.e. the cost of investment per resources increases with the amount of resources. For simplicity in analysis, we consider these costs to be  quadratic\footnote{The overall intuitions of the model are expected to be the same in the case of considering any convex function of $I_F$.}. More specifically, we consider the fee that SP$_F$ pays to SP$_L$ for resources to be  $sI^2_F$, where $I_F$ is the number of resources that SP$_F$ reserves from SP$_L$. We also consider the investment cost that SP$_L$ incurs to be $\gamma I^2_L$, where $I_L$  is the number of resources that SP$_L$ adds to her infrastructure, and $\gamma$ is the marginal cost of investment. Thus, the payoffs of SPs are:
\begin{equation} \label{equ:payoffF}
\pi_F=n_F(p_F-c)-sI^2_F
\end{equation}
\begin{equation}\label{equ:payoffL}
\pi_L=n_L(p_L-c)+sI^2_F-\gamma I^2_L
\end{equation}
%
%

Naturally, we assume that the fee per resources ($s$) that SP$_F$ pays to SP$_L$ is at least equal to the marginal cost of investment on the infrastructure ($\gamma$), i.e. $s\geq \gamma$. Also, note that $I_F\leq I_L$. To have a non-trivial problem, we also assume $I_L>0$ and $I_F\geq 0$. 

The strategies of SP$_L$ and SP$_F$ is to choose the access fee for EUs ($p_L$ and $P_F$, respectively) and the level of investment ($I_L$ and $I_F$, respectively). First, we assume that the per resource fee $s$ is pre-determined. In effect, we relax this assumption later when we discuss about the bargaining framework  between SP$_L$ and SP$_F$. 

\subsection*{EUs:}  \edit{The strategy of an EU is to choose one of the SPs to buy wireless plan from. Towards that end, we consider that SPs decide based on their utility which depends on the 
resources the SPs invest (which is an indicator of the quality they provide), the inherent preferences for SPs\footnote{These preferences do not depend on other parameters of the model and strategies. This is why we denote them as \emph{inherent} preferences. As discussed in the introduction, these preferences can be because of different services that SPs offer, e.g. bundling wireless service with a free or cheap international call plan. Moreover, the preferences can be because of the reluctance of  an EU to buy her wireless plan from a particular SP, e.g. an  SP with an infamous customer service.}, and the also the price they pay. In particular, the utility of an EU with preference $x_j$  for  SP$_{j}$, $j\in\{L,F\}$ is:
\begin{equation}\label{equ:EUutility}
u_{x_j}=v^*-t_jx_j-p_j
\end{equation}
where $t_j$ is depends on the relative investment level of the SPs (explained later), and $v^*$ is a common valuation that captures the value of buying a wireless plan for EUs regardless of the SP chosen, the investment level, or the price.}

To model the inherent preferences of EUs, we use a hotelling model. We assume that  SP$_L$ is located at 0, SP$_F$ is located at 1, and EUs are distributed uniformly along the unit interval $[0,1]$. The closer an EU to an SP, the more this EU prefers this SP to the other. Specifically, the EU located at $x\in[0,1]$ incurs  a cost of $t_L x$  (respectively, $t_F(1-x)$) when joining SP$_L$ (respectively, SP$_F$). Note that the notion of closeness and distance is used to model the preference of EUs, and may not be the same as  physical distance.  


We now describe how $t_j$'s    depend on the relative investment level of  SPs. We consider that the higher the investment level of SP$_L$ (respectively, SP$_F$) in comparison to the other SP, the lower would be $t_L$ (respectively $t_F$). Thus:
\begin{equation}\label{equ:t_L}
t_L=\frac{I_F}{I_L}
\end{equation} 
\begin{equation}\label{equ:t_F}
t_F=1-t_L=\frac{I_L-I_F}{I_L}
\end{equation}
Thus, $t_L$ and $t_F$ capture the impact of quality of service in the decision of EUs.

\edit{Note that from \eqref{equ:EUutility}, even when $I_F=0$ (i.e. SP$_F$ makes zero investment) and $I_F=I_L$ (i.e. EUs of SP$_L$ do not get any benefit from the investment of SP$_L$), EUs may still join SP$_F$ and SP$_L$, respectively. This happens since both SP$_F$ and SP$_L$ usually have access to separate resources that they can use to serve the EUs who join them, above and beyond the strategic investments $I_F, I_L$ they make and acquire for their mutual transactions.  For example, SP$_L$ will have additional spectrum and infrastructure beyond what it offers to SP$_F$ (i.e, $I_L$), and SP$_F$ would also have spectrum and infrastructure it acquires from other MNOs (many MVNOs simultaneously transact with multiple MVNOs). The quality of service that can be obtained from these additional resources  is captured in $v^*$.
}

We also assume the full market coverage of EUs by SPs. This means that  each EU chooses exactly one SP to subscribe to.  This  assumption is common in hotelling models and is necessary to ensure competition between SPs. An equivalent assumption is to consider the common valuation $v^*$ to be sufficiently large so that the utility of EUs for buying a wireless plan  is positive regardless of the choice of SP. 
\edit{Later, in Section~\ref{section:generalization}, we consider the demand of SPs to be generated both through the hotelling model and  also a demand function. } This, in effect, models an outside option for EUs and relaxes the assumption of  full market coverage of EUs by SPs. 

\subsection*{Non-Cooperative Sequential Game Formulation:}
We model the interaction between SPs and EUs using a four-stage sequential game. Naturally, we assume that SP$_L$ makes the first move and is the leader of the game. The timing and the stages of the game are as follows:
\begin{enumerate}
\item SP$_{L}$ decides on the investment on the infrastructure ($I_L$).
\item SP$_{F}$ decides on the investment, i.e. number of resources to lease from SP$_L$ ($I_F$ such that $I_F\leq I_L$).
\item SP$_{L}$ and SP$_{F}$ determine the access fees for EUs (respectively, $p_L$ and $p_F$).
\item EUs decide to subscribe to one of the SPs.
\end{enumerate}

We assumed the selection of investments by SPs ($I_L$ and $I_F$)  happens before the selection of prices for EUs ($p_L$ and $p_F$) since investment decisions are long-term decision. These decisions are expected to be kept constant over longer time horizons in comparison to pricing decisions. 

In the sequential game framework,  we seek a \emph{Subgame Perfect Nash Equilibrium}  using \emph{backward induction}:

\begin{definition}\emph{Subgame Perfect Nash Equilibrium (SPNE):}
	A strategy is an SPNE if and only if it constitutes a Nash Equilibrium (NE) of every subgame of the game. 
\end{definition} 


 


\section{\edit{EUs with No Outside Option}}\label{section:analysis}
In this section, we consider that EUs are modeled using the hotelling model. First, we characterize the SPNE strategies  of the game in Section \ref{section:SPNEanalysis}. Then, we summarize and discuss about the outcome of the game in Section~\ref{section:outcome}. We provide numerical results in Section~\ref{section:numerical}. The bargaining framework for this case is presented in Section~\ref{section:bargaining}.

\subsection{SPNE Analysis}\label{section:SPNEanalysis}
We use the backward induction to characterize SPNE strategies. This means that we  characterize the equilibrium strategies starting from the last stage of the game and proceeding backward.  Thus, we start with Stage 4:

\subsubsection*{Stage 4} In this subsection, we characterize the division of EUs between SPs in the  equilibrium, i.e. $n_L$ and $n_F$, using the knowledge of the strategies chosen by the  SPs in Stages 1, 2, and 3. To do so, we characterize the location of the EU that is indifferent between joining either of the SPs, $x_n$. Thus, EUs located at $[0,x_n)$ join SP$_L$, and those located at $(x_n,1]$ joins the non-neutral SP$_F$ (using the full market coverage assumption). 

The EU located at $x_n\in[0,1]$ is indifferent between connecting to SP$_L$ ans SP$_F$ if:
\begin{equation}\label{equ:xn_t_parameter}
\small
	\begin{aligned}
	&v^*-t_{F}(1-x_n)-p_{F}=v^*-t_Lx_n-p_L\\
	&\qquad \qquad \qquad  \Rightarrow x_n=\frac{t_F+p_F-p_L}{t_L+t_F}
	\end{aligned}
\end{equation}

Note that $t_F+t_L=1$. Substituting the value of $t_F$ \eqref{equ:t_F} yields:
\begin{equation}\label{equ:xn}
\small
x_n=\frac{I_L-I_F}{I_L}+p_F-p_L
\end{equation}
Note that EUs in the interval $[0,x_n)$ joins SP$_L$ and  those in the interval $(x_n,1]$ joins SP$_F$. Thus, the fraction of EUs with each SP ($n_L$ and $n_F$) is:
\begin{equation} \label{equ:EUs_stage4}
\small
\begin{aligned} 
n_L &=
\left\{
\begin{array}{ll}
	0  & \mbox{if } x_n < 0 \\
	\frac{I_L-I_F}{I_L}+p_F-p_L & \mbox{if } 0\leq x_n \leq 1\\
	1 & \mbox{if } x_n>1
\end{array}
\right. \\
n_F&=1-n_L
\end{aligned}
\end{equation}

\subsubsection*{Stage 3}

In this stage, SP$_L$ and SP$_F$ determine their  prices for EUs, $p_L$ and $p_F$, respectively, to maximize their payoff. We seek  Nash equilibrium (NE) strategies. Note that in general there might exist several NE strategies: some of them corner equilibria (an extreme case in which one of the SPs receives zero EUs) and some interior equilibria (in which both SPs receive a positive mass of EUs). We first argue, that there is no corner equilibria. Then, we characterize the unique interior equilibrium.

\begin{theorem}\label{theorem:NoCorner}
There is no corner equilibrium, i.e. there is no equilibrium by which $n_L=0$ or $n_F=0$.
\end{theorem}
Proof is presented in Appendix~\ref{section:appendix:corner}.

Now, we look for all NE by which  $0< x_n<1$ ($x_n$ characterized in the previous stage of the game). In this case, using \eqref{equ:payoffF}, \eqref{equ:payoffL}, and \eqref{equ:EUs_stage4}, the payoffs of SPs would be:
\begin{equation}\label{equ:payoffF_2}
\small
\pi_F=(\frac{I_F}{I_L}+p_L-p_F)(p_F-c)-sI^k_F 
\end{equation}
\begin{equation}\label{equ:payoffL_2}
\small
\pi_L=(\frac{I_L-I_F}{I_L}+p_F-p_L)(p_L-c)+sI^k_F-\gamma I^k_L
\end{equation}
\normalsize
In the following theorem, we prove that the NE uniquely exists and  characterize it:
\begin{theorem}\label{theorem:stage 3}
The NE strategies of Stage 3 by which  $0<n_L<1$ and $0<n_F<1$ are unique, and are $p_F=c+\frac{I_L+I_F}{3I_L}$ and $p_L=c+\frac{2I_L-I_F}{3I_L}$. 
\end{theorem}

\begin{remark}\label{remark:price}
Note that $p_L$ (respectively, $p_F$) is increasing with respect to $I_L$ (respectively, $I_F$). Also, $p_L$ (respectively, $p_F$) is decreasing with respect to $I_F$ (respectively, $I_L$). Also note that prices depend only on the cost and the relative investment level of the SPs, i.e. $\frac{I_F}{I_L}$.
\end{remark}

\begin{proof}
In this case,  for an interior NE, $p^*_F$ and $p^*_L$ should be determined to  satisfy the first order condition, i.e. $\frac{d \pi_F}{dp_F}=0$ and   $\frac{d \pi_L}{dp_L}=0$.  The first order conditions yield:
$$
\small
\begin{aligned} 
2p^*_F-p^*_L&=\frac{I_F}{I_L}+c \qquad \& \qquad 2p^*_L-p^*_F&=\frac{I_L-I_F}{I_L}+c 
\end{aligned}
$$
Thus,
\begin{eqnarray}\label{equ:optimump}
\small
\begin{aligned}
p^*_F=c+\frac{I_L+I_F}{3I_L} \qquad \& \qquad p^*_L=c+\frac{2I_L-I_F}{3I_L} 
\end{aligned}
\end{eqnarray}
\normalsize
Therefore, $p^*_F$ and $p^*_L$ are the unique NE strategies if they yield $0< x_n< 1$ and no unilateral deviation is profitable for SPs. In Part A, we check the former, and in Part B, we check the latter. 

\textbf{Part A:}
We first  check the condition that with $p^*_L$ and $p^*_F$, $0< x_n< 1$. Using \eqref{equ:xn} and \eqref{equ:optimump}:
$$
\small
x_n=\frac{I_L-I_F}{I_L}+p^*_F-p^*_L= \frac{2I_L-I_F}{3I_L}
$$
which is greater than zero since $I_L\geq I_F$ and $I_L>0$. $x_n$ is also clearly less than one (note that $I_F\geq 0$). Thus, the condition $0< x_n< 1$ holds.

\textbf{Part B:}
Note that $\frac{d^2 \pi_F}{dp^2_F}<0$ and   $\frac{d^2 \pi_L}{dp^2_L}<0$, at all values of $p_F$ and $p_L$. Thus, a local maxima would be a global maximum, and any solutions to the first order conditions would maximize the payoff of the SPs when $0<x_n<1$, and no unilateral deviation by which $0<x_n<1$ would be profitable for SPs. 

Now, we discuss that any deviation by SPs such that $n_L=0$ and $n_L=1$ (which subsequently yields $n_F=1$ and $n_F=0$, respectively) is not profitable. Note that the payoff of SPs, \eqref{equ:payoffF} and \eqref{equ:payoffL}, are continuous  as $n_L\downarrow 0$, and $n_L\uparrow 1$ (which subsequently yields  $n_F\uparrow 1$ and $n_F\downarrow 0$, respectively). Thus, the payoffs of both SPs when selecting $p^*_L$ and $p^*_F$ (solutions of first order conditions) are greater than or equal to the payoffs when $n_L=0$ and $n_L=1$. Thus, any deviation by SPs such that $n_L=0$ or $n_L=1$ is not profitable for SPs. 


This complete the proof that $p^*_F$ and $p^*_L$ are the unique NE strategies by which both SPs are active, i.e. $0<x_n<1$.  
\end{proof}
 
\subsubsection*{Stage 2}

In this stage of the game, SP$_F$ decides on the investment level, i.e. the number of resources to be leased from SP$_L$ ($I_F$), with the condition that $I_F\leq I_L$ to maximize $\pi_F$. The optimization that would be solved by SP$_F$ is:

\begin{equation} \label{equ:maxpiFIF}
\max_{I_F\leq I_L}\pi_F=\max_{I_F\leq I_L}\big{(}\frac{I_L+I_F}{3I_L}\big{)}^2-sI^2_F
\end{equation} 
Note that for the last equality, we used  \eqref{equ:payoffF_2} and Theorem~\ref{theorem:stage 3}. 

\begin{theorem}\label{theorem:stage2}
The optimum investment level of SP$_F$ is:
\begin{equation}\label{equ:optimumI_F}
I^*_F=\left\{
		\begin{array}{ll}
			\frac{I_L}{9sI^2_L-1} & \mbox{if } \quad I_L>\sqrt{\frac{2}{9s}}\\
			 I_L & \mbox{if } \quad  I_L\leq \sqrt{\frac{2}{9s}}
		\end{array}
	\right.
	\end{equation}
\end{theorem}

\begin{remark}\label{remark:continuousI_L}
Note that in \eqref{equ:optimumI_F}, $I^*_F$ is continuous at $I_L=\sqrt{\frac{2}{9s}}$. 
\end{remark}

\begin{remark}\label{remark:I_L_dependencies}
If the fee per resources ($s$) or the investment by SP$_L$ ($I_L$) is low, then SP$_F$ reserves all the available resources. If not, then SP$_F$ reserves a fraction of available resources ($I^*_F<I_L$). Note that in this case, $\frac{d I^*_F}{d I_L}<0$. Thus, the higher the number of available resources, the lower would be the number of resources reserved by SP$_F$. In addition, $I^*_F$ is decreasing with $s$. 
\end{remark}

\begin{proof}
Note that in \eqref{equ:maxpiFIF}, $\pi_F$ is concave over $I_F$ if $s>\frac{1}{9I^2_L}$, $\pi_F$ is convex if $s<\frac{1}{9I^2_L}$, and is linear if $s=\frac{1}{9I^2_L}$.  We characterize the optimum investment level, i.e. $I^*_F$ under each of these conditions:

\subsection*{$s>\frac{1}{9I^2_L}$:}  In this case, the payoff of SP$_F$ is concave. Thus, the first order condition yields the possible optimum investment level, $\hat{I}_F$:
\begin{equation} \label{equ:hatI_f}
\begin{aligned}
\small
\frac{d \pi_F}{d I_F}|_{\hat{I}_F}=0&\Rightarrow \frac{2(\hat{I}_F+I_L)}{9I^2_L}-2s\hat{I}_F=0\\
&\Rightarrow \hat{I}_F=\frac{I_L}{9sI^2_L-1}
\end{aligned}
\end{equation}
Thus, the optimum investment level, $I^*_F$ is:
\begin{equation}
I^*_F=\min\{\hat{I}_F,I_L\} 
\end{equation}
Note that from the assumption of this case $9sI^2_L>1$. Thus, $\hat{I}_F>0$ (using \eqref{equ:hatI_f}). In addition, from the expression of $\hat{I}_F$ \eqref{equ:hatI_f}:
$$
\hat{I}_F<I_L\iff 9sI^2_L>2
$$ 
Thus, the optimum investment level for SP$_F$ would be:
\begin{equation}
\small
I^*_F=\left\{
		\begin{array}{ll}
			\frac{I_L}{9sI^2_L-1} & \mbox{if } \quad I_L>\sqrt{\frac{2}{9s}}\\
			 I_L & \mbox{if } \quad  I_L\leq \sqrt{\frac{2}{9s}}
		\end{array}
	\right.
\end{equation}
\subsection*{$s<\frac{1}{9I^2_L}$:} In this case, $\pi_F$ would be convex. Thus, the optimum level of investment would be on the boundaries of the feasible set. Thus, $I^*_F$ would be either 0 or $I_L$, whichever yields a higher payoff. Note that:
$$
\begin{aligned}
\pi_F|_{I_F=0}&= \frac{1}{9}\qquad \& \qquad \pi_F|_{I_F=I_L}&=\frac{4}{9}-sI^2_L
\end{aligned}
$$
Thus, $I^*_F=I_L$ if and only if $\frac{4}{9}-sI^2_L\geq \frac{1}{9}$ (Note that we assumed that if $I_F=0$ and $I_F=I_L$ yield the same payoff, then SP$_F$ chooses $I_F=I_L$). Thus, $s\leq \frac{1}{3I^2_L}$ yields $I^*_F=I_L$. Note that from the assumption of the case,  $s<\frac{1}{9I^2_L}<\frac{1}{3I^2_L}$. Thus, $I^*_F=I_L$.

\subsection*{$s=\frac{1}{9I^2_L}$:} In this case, $\pi_F$ would be an increasing linear function of $I_F$. Thus, $I^*_F=I_L$. 

Putting all the cases together, the result of the theorem follows.
\end{proof}

\subsubsection*{Stage 1} In this stage, SP$_L$ decides on the level of investment $I_L$ to maximize her payoff, $\pi_L$, where $\pi_L$ can be characterized using   \eqref{equ:payoffL_2} and Theorem \ref{theorem:stage 3}:
\begin{equation}\label{equ:optimizationpil}
\max _{I_L}\pi_L=\max_{I_L} (\frac{2I_L-I^*_F}{3I_L})^2+sI^{*2}_F-\gamma I^2_L
\end{equation}
Note that $I^*_F$ is characterized in \eqref{equ:optimumI_F}. 

Now, we present a tie-breaking assumption by which SP$_L$ decides on the level of investment chosen when there exists multiple optimum solutions to \eqref{equ:optimizationpil}.\edit{ This assumption is motivated by the fact that SPs prefer to offer a better quality of service provided that their payoffs  do not suffer by doing so:}
\begin{assumption}\label{assumption:tie_IL}
If there exists multiple investment levels, i.e. $I_L$, by which \eqref{equ:optimizationpil} is maximized, then in the equilibrium, SP$_L$ chooses the largest investment level.
\end{assumption}

In the next theorem, we characterize the candidate optimum answers:

\begin{theorem}\label{theorem:stage 1}
The optimum solution to \eqref{equ:optimizationpil} chosen by SP$_L$ is $I^*_L=\max(\mathbf{\hat{I}_L})$, where $\mathbf{\hat{I}_L}$ is the set of all solutions to the following optimization problem:
\begin{equation}
\small
\max_{I_L\geq \sqrt{\frac{2}{9s}}}{\pi_{L,I^*_F}=\frac{1}{9}(2-\frac{1}{9sI^2_L-1})^2+s(\frac{I_L}{9sI^2_L-1})^2-\gamma I^2_L}
\end{equation}
\end{theorem}

\begin{remark}
Theorem yields that the minimum optimum level of investment by SP$_L$ is $\sqrt{\frac{2}{9s}}$. 
\end{remark}

\begin{proof}
To find the optimum level of investment, $I^*_L$, consider two scenarios:
\subsection*{$I_L\leq \sqrt{\frac{2}{9s}}$:} In this case, $I^*_F=I_L$ by \eqref{equ:optimumI_F}. Thus, the maximization \eqref{equ:optimizationpil} would become:
$$
\max_{I_L\leq \sqrt{\frac{2}{9s}}}(s-\gamma) I^2_L
$$
Thus:
\begin{equation}
I^*_L=\left\{
		\begin{array}{ll}
			\sqrt{\frac{2}{9s}} & \mbox{if } \quad s\geq \gamma\\
			 0 & \mbox{if } \quad  s< \gamma
		\end{array}
	\right.
\end{equation}
Note that Assumption \ref{assumption:tie_IL} is the reason that for $s=\gamma$, $I_L=\sqrt{\frac{2}{9s}}$ is chosen. Also, recall that to have a non-trivial problem, we assumed $s\geq \gamma$. Thus, if $I_L\leq \sqrt{\frac{2}{9s}}$, then $I^*_L=\sqrt{\frac{2}{9s}}$. Thus, this case can be considered as part of the next part.

\subsection*{$I_L\geq \sqrt{\frac{2}{9s}}$:} 
In this case, $I^*_F=\frac{I_L}{9sI^2_L-1}$. Note that $I^*_F$ \eqref{equ:optimumI_F} is continuous at $I_L=\sqrt{\frac{2}{9s}}$ (Remark~\ref{remark:continuousI_L}). Thus, as $I_L\rightarrow \sqrt{\frac{2}{9s}}$, $\pi_{L,I^*_F}\rightarrow \pi_L|_{I_L=\sqrt{\frac{2}{9s}}}$. Therefore, this case also includes the optimum answer of the previous case. Thus, the maximization \eqref{equ:optimizationpil}   would become:
\begin{equation}
\small
\max_{I_L\geq \sqrt{\frac{2}{9s}}}{\pi_{L,I^*_F}=\frac{1}{9}(2-\frac{1}{9sI^2_L-1})^2+s(\frac{I_L}{9sI^2_L-1})^2-\gamma I^2_L}
\end{equation}
Note that using Assumption~\ref{assumption:tie_IL}, if there exists multiple solutions to the maximization, the highest one would be chosen by SP$_L$. Results of the theorem follows.
\end{proof}

\subsection{The Outcome of the Game} \label{section:outcome}
In this section, we characterize the equilibrium outcome of the game using the results of the previous section and discuss about them. From Theorem~\ref{theorem:stage 1}, $I^*_L$ can be either $\sqrt{\frac{2}{9s}}$ or greater than that.  In Corollaries~\ref{corollary:outcome_IL>} and \ref{corollary:outcome_IL<}, we characterize
the equilibrium outcome when $I^*_L>\sqrt{\frac{2}{9s}}$  and $I^*_L\leq \sqrt{\frac{2}{9s}}$, respectively. Note that in each case, there exists a unique outcome and the two cases are mutually exclusive. Thus, the outcome of the game exists and is unique. 

\subsubsection{If $I^*_L>\sqrt{\frac{2}{9s}}$}
\begin{corollary} [Outcome A]\label{corollary:outcome_IL>}
If $I^*_L>\sqrt{\frac{2}{9s}}$, then:
\begin{itemize}
\item\textbf{Stage 1:}  The optimum level of investment by SP$_L$ is characterized in  Theorem~\ref{theorem:stage 1}.
\item \textbf{Stage 2:} The optimum level of investment by SP$_F$ is: $I^*_F=\frac{I^*_L}{9sI^{* 2}_L-1}$.
\item \textbf{Stage 3:} Prices for EUs are: $p^*_F=c+\frac{1+\frac{1}{9sI^2_L-1}}{3}$ and $p^*_L=c+\frac{2-\frac{1}{9sI^2_L-1}}{3}$. 
\item \textbf{Stage 4:} The fraction of EUs with each SP are: $n^*_F=\frac{1+\frac{1}{9sI^2_L-1}}{3}$ and $n^*_L=\frac{2-\frac{1}{9sI^2_L-1}}{3}$.
\end{itemize}
\end{corollary}
\begin{proof}
The outcome of Stage 1 directly follows from the results of Stage 1 (Theorem~\ref{theorem:stage 1}). The outcome of Stage 2 directly follows from the results of Stage 2 (Theorem~\ref{theorem:stage2}) and the assumption of corollary that $I^*_L>\sqrt{\frac{2}{9s}}$. The outcome of Stage 3 follows from \eqref{equ:optimump} and the outcome of Stage 2. The outcome of Stage 4 follows from  \eqref{equ:EUs_stage4} and the outcomes of Stage 2 and 3. 
\end{proof}

%
%

Note that in this case, as discussed, $I^*_F$ is decreasing with respect to $I^*_L$ and $s$ (Remark~\ref{remark:I_L_dependencies}). In addition,  when $s$ is fixed, $p^*_F$ and $p^*_L$ are decreasing and increasing with $I^*_L$, respectively.  Also, when $I_L$ is fixed, $p^*_F$ and $p^*_L$ are decreasing and increasing with respect to the per resource fee, $s$, respectively. \edit{Note that $s$ also affects the selection of $I_L$. Thus, the general relation between prices and $s$ (when $I_L$ is not fixed) is more complicated , and is investigated through numerical analysis in Section~\ref{section:numerical}.} The trends in $n^*_F$ and $n^*_L$ are similar to those of $p^*_F$ and $p^*_L$, respectively. 


 In  Section~\ref{section:numerical}, we  observe that $I^*_L$ is dependent on $s$. Thus, the relationship between $s$ and the outcome is more complicated.

\subsubsection{If $I^*_L\leq \sqrt{\frac{2}{9s}}$}

\begin{corollary}[Outcome B]\label{corollary:outcome_IL<}
If $I^*_L\leq  \sqrt{\frac{2}{9s}}$, then:
\begin{itemize}
\item \textbf{Stage 1:} $I^*_L= \sqrt{\frac{2}{9s}}$.
\item \textbf{Stage 2:} $I^*_F=I^*_L$.
\item \textbf{Stage 3:} $p^*_F=c+\frac{2}{3}$ and $p^*_L=c+\frac{1}{3}$.
\item \textbf{Stage 4:}  $n^*_F=\frac{2}{3}$ and $n^*_L=\frac{1}{3}$.
\end{itemize}
\end{corollary}
Proof is similar to the proof of the previous corollary. 

In this case,  SP$_F$ reserves all available resources, and  the investment level of SP$_L$ (which is equal to the number of resources reserved by SP$_F$) is a decreasing function of the fee per resources, i.e. $s$. SP$_F$ quotes a higher price for EUs in comparison to SP$_L$. In spite of the higher price, SP$_F$ would be able to attract more EUs given the better investment level in comparison to SP$_L$ which translates into a better preference for EUs. 

\edit{Note that as mentioned in Remark~\ref{remark:price}, the prices of (and subsequently the fraction of EUs with)  SPs depend on the relative investment level of SPs. Thus, since in outcome A, the relative investment is fixed, i.e. $\frac{I_F}{I_L}=1$, the prices (and subsequently the fraction of EUs) are constant with respect to the investment values and subsequently $s$. Moreover, note that although SP$_F$ leases all the new resources of SP$_L$, SP$_L$ is still able to attract $\frac{1}{3}$ of EUs. This is because EUs obtain benefits associated with SPs regardless of their strategies, as mentioned in Model (Section~\ref{section:model}) in discusson after \eqref{equ:t_F} and \eqref{equ:t_L}.}

We  calculate $\pi_L$ and $\pi_F$ using Corollary~\ref{corollary:outcome_IL<} and \eqref{equ:payoffL}:
\begin{equation}
\begin{aligned}
\pi^*_L&=\frac{1}{3}-\frac{2 \gamma}{9 s}\qquad \& \qquad \pi^*_F&=\frac{2}{9}
\end{aligned}
\end{equation}
Thus, the payoff of SP$_L$ would be higher than the payoff of SP$_F$, i.e. $\pi^*_L>\pi^*_F$, if and only if $s>2\gamma$. In this case, although in comparison to SP$_F$, SP$_L$ offers her service to EUs with lower price and  attracts a lower number of EUs, she can still obtain a higher payoff through the per resource fee that she collects from SP$_F$. Thus, SP$_L$ leases the resources to SP$_F$ and instead generates revenue through the fee she charges to SP$_F$.

\subsection{Numerical Results  and Discussions}  \label{section:numerical}

In this section, we use numerical simulations (i) to determine whether and under what conditions the outcomes in Corollaries~\ref{corollary:outcome_IL>} (Outcome A) and \ref{corollary:outcome_IL<} (Outcome B) would emerge, and (ii)  to provide insights for results under different parameters of the model. For all results, we consider $c=1$.\footnote{Note that the choice of $c$ barely affects the results. It may only shift some of the results (e.g. the price charged to EUs) by only a constant.}

\begin{figure}[t]
	\begin{subfigure}{.25\textwidth}
		\centering
		\includegraphics[width=\linewidth]{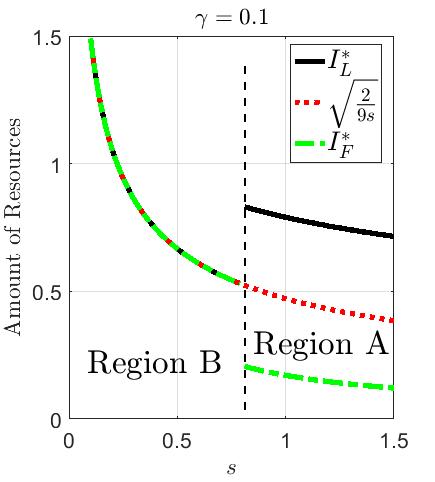}
		\label{figure:resources_gamma0.1}
	\end{subfigure}%
	\begin{subfigure}{.25\textwidth}
		\centering
		\includegraphics[width=\linewidth]{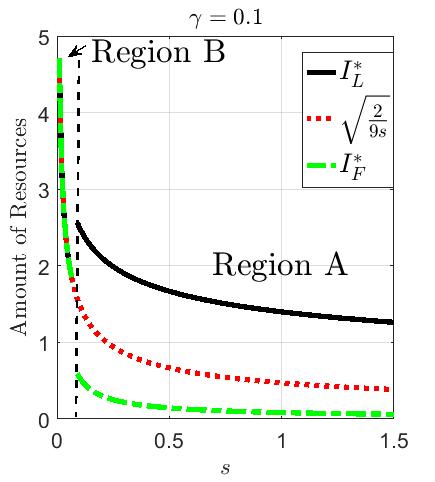}
		\label{fig:resources_gamma0}
	\end{subfigure}
	\caption{Investment decisions of SPs vs. per resource fee, $s$. Note that in Region B, $I^*_F=I^*_L=\sqrt{\frac{2}{9s}}$.}\label{figure:resources}
\end{figure}

In Figure~\ref{figure:resources}, we plot the optimum level of investment of SP$_L$ ($I^*_L$), the number of resources that SP$_F$ reserves ($I^*_F$), and the minimum optimum level of investment by SP$_L$ ($\sqrt{\frac{2}{9s}}$), when $\gamma=0.1$ (left) and $\gamma=0.01$ (right) (recall that,  in \eqref{equ:payoffL}, $\gamma$ is the marginal cost of investment.). The discontinuities  are because of the transition of the outcome of the game from the Outcome B to Outcome A. Thus, when $s$ is smaller than a threshold, SP$_L$ invests minimally ($I^*_L=\sqrt{\frac{2}{9s}}$), and SP$_F$ leases all the new resources of SP$_L$. On the other hand, after that threshold, SP$_L$ invests more than the minimum, and SP$_F$ leases only a portion of the new resources invested by SP$_L$. The amount that SP$_F$ leases in this region is smaller than the amount she leases in the region of Outcome B.  Results also reveal that   $I^*_L$ is strictly decreasing  in  each outcome with $s$. Thus, the higher the fee per resources that SP$_F$ pays to SP$_L$, the lower would be the number of resources that SP$_F$ reserves (as revealed by results in Theorem~\ref{theorem:stage2}), and  the lower  would be (in each region) the investment of SP$_L$.  Also, when the marginal cost of investment is small ($\gamma=0.01$),  SP$_L$ investment is  more than the case that $\gamma=0.1$. Thus, as expected, the higher the investment cost, the lower would be the investment.



In Figure~\ref{figure:prices}, we plot the pricing decisions of SPs for EUs, i.e. $p^*_L$ and $p^*_F$, for $\gamma=0.1$ (left) and $\gamma=0.01$ (right). Similar to Figure~\ref{figure:resources},  the discontinuities in the figures are because of the transition of the outcome of the game from Outcome B to A.  Note that results in Corollary~\ref{corollary:outcome_IL<} revealed that in Outcome B (when $s$ is sufficiently small), $p^*_F$ and $p^*_L$ are constant (given that c is a constant) independent of $\gamma$ and $s$ and $p^*_F>p^*_L$. On the other hand, in Region A, SP$_L$ would be able to charge a higher price than SP$_F$ (Figure~\ref{figure:prices}-right). This is due to the fact that in Region A, SP$_L$  leases only a portion of the new resources to SP$_F$ (Figure~\ref{figure:resources}) and uses the rest for her own EUs. Furthermore, in this region, $p^*_F$ is decreasing and $p^*_L$ is increasing with respect to $s$. Note that the number of resources for both SP$_L$ and SP$_F$ is decreasing with $s$. However, this decreasing behavior affect the price that SPs  charge to EUs differently: makes SP$_F$ to decrease her price and enables SP$_L$ to increase her price.


\begin{figure}[t]
	\begin{subfigure}{.25\textwidth}
		\centering
		\includegraphics[width=\linewidth]{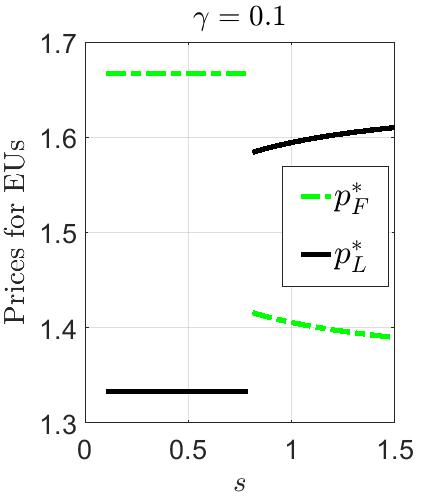}
		\label{fig:prices_gamma01_2}
	\end{subfigure}%
	\begin{subfigure}{.25\textwidth}
		\centering
		\includegraphics[width=\linewidth]{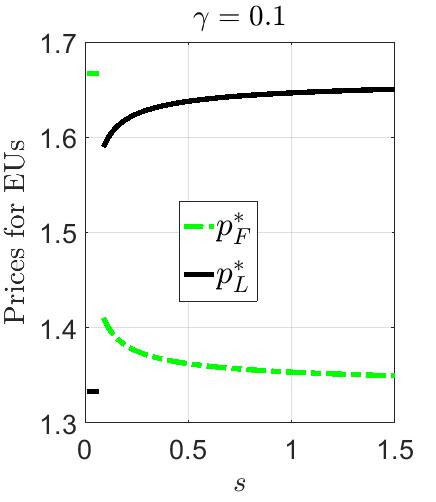}
		\label{fig:prices_gamma0}
	\end{subfigure}
	\caption{Pricing decisions of SPs for EUs versus $s$}\label{figure:prices}
\end{figure}


Note that by Corollaries \ref{corollary:outcome_IL>} and \ref{corollary:outcome_IL<}, the dependency of $n^*_L$ and $n^*_F$ to parameters of the model is  similar to the dependency of $p^*_L$ and $p^*_F$. The only difference is the exclusion of $c$ from the expressions.

\subsection{The Bargaining Framework}\label{section:bargaining}
\edit{In this section, we complement the previous non-cooperative framework by considering a hybrid of cooperative  and non-cooperative scenarios. In this case, the MVNO and the MNO jointly decide on the investments ($I_L$, $I_F$), but individually decide on the prices for EUs. They split the profit (by selecting the per resource fee $s$) based on the number of EUs  and profit  each acquire.}

In particular, for the cooperative part, we formulate a bargaining framework, and use  the \emph{Nash Bargaining Solution} (NBS) to characterize $I_F, I_L$, and $s$:

 \begin{definition}\emph{Nash Bargaining Solution (NBS):} is the unique solution (in our case the tuple of the payoffs of SP$_L$ and SP$_F$) that satisfies the four ``reasonable" axioms (Invariant to affine transformations, Pareto optimality, Independence of irrelevant alternatives, and Symmetry) characterized in \cite{osbornebargaining}.
\end{definition}

Note that after characterizing the NBS, the access fee for EUs ($p_L$ and $p_F$) and the split of EUs between SPs can be determined from Theorem~\ref{theorem:stage 3} and equation \eqref{equ:EUs_stage4}.


In this cooperative framework, a collusive outcome may occur in which both SPs jointly decrease the  amount of investment while maintaining a specific relative difference, i.e. maintaining the differentiation between them that yields the best outcome while minimizing the investment. The reason is that  EUs decide based on the ratio of the investment by SPs and not the absolute values, i.e. EUs decide relatively. Thus, a regulatory intervention may be desirable. Therefore, for the rest of this section, we consider a regulator that enforces a minimum investment level of $I_{L,min}$ on SP$_L$, i.e. $I_L\geq I_{min,L}$.

Let $0\leq w\leq 1$ be the relative bargaining power of the SP$_F$ over SP$_L$: the higher $w$, more powerful is the bargaining power of SP$_F$. In addition, $\pi_{L}$ and $\pi_{F}$ denote the payoff of the CP and SP respectively, and $d_{L}$ and $d_{F}$ denote the payoff each decision maker receives in case of disagreement, i.e. \emph{disagreement payoff}. In order to characterize the disagreement payoffs, we assume that in case of disagreement, the SPs will play the sequential game whose outcome characterized previously in Section~\ref{section:outcome}. 
 Note that the value of the disagreement payoff for the SPs can have an effect similar to the bargaining power ($w$ for SP$_F$ and $1-w$ for the SP$_L$). 

Using standard game theoretic results in \cite{osbornebargaining}, the pair of $\pi^*_{L}$ and $\pi^*_{F}$ can be identified as the Nash bargaining solution of the problem if and only if it solves the following optimization problem:
\begin{equation}\label{equ:nash_solution_1}
\begin{aligned}
&\max_{\pi_L, \pi_F}{(\pi_F-d_F)^{w}(\pi_L-d_L)^{1-w}}\\
& \text{s.t.}\\
&\qquad (\pi_L,\pi_F)\in U\\
& \qquad (\pi_L,\pi_F)\geq (d_L,d_F)
\end{aligned}
\end{equation}
where $U$ is the set of feasible payoff pairs, and $\pi_F$ and $\pi_L$ are of the form of \eqref{equ:payoffF} and \eqref{equ:payoffL}. 

Since SPs bargain over $s$, $I_L$, and $I_F$, the maximization~\eqref{equ:nash_solution_1}, should be over $s$, $I_L>I_{min,L}$, and $I_F\leq I_L$.  Thus the maximization is,
\begin{equation}\label{equ:nash_solution_2}
\begin{aligned}
&\max_{s, I_L, I_F}{(\pi_{F}-d_{F})^{w}(\pi_{L}-d_{L})^{1-w}}\\
& \text{s.t.}\\
&\qquad 0\leq I_F\leq I_L \\
&\qquad I_L\geq I_{min,L}\\
&\qquad \pi_{L}\geq d_{L}\\
&\qquad \pi_{F}\geq d_{F}
\end{aligned}
\end{equation}

We define $s^*$, $I^*_L$, and $I^*_F$ to be the optimum solution of \eqref{equ:nash_solution_2}. Note that these parameters combined with results in Theorem~\ref{theorem:stage 3} and Equation~\eqref{equ:EUs_stage4} characterize the optimum  division of profit ($\pi^*_L$ and $\pi^*_F$) and thus the NBS.

 In addition, we define the \emph{aggregate excess profit} to be the additional profit yielded from the cooperation in the bargaining framework:
\begin{definition}\emph{Aggregate Excess Profit ($u_{excess}$):} The aggregate excess profit is defined as follows:

	\begin{equation}\label{equ:uexcess}
	\begin{aligned}
	u_{excess}&=\pi_{L}-d_{L}+\pi_{F}-d_{F}\\
	&=n_L(p_L-c)-\gamma I^2_L-d_{L}+n_F(p_F-c)-d_{F}
\end{aligned}	
	\end{equation}
	
\end{definition}
Note that $u_{excess}$ is independent of $s$ and is only a function of $I_L$ and $I_F$ (since $n_L$, $n_F$, $p_L$, and $p_F$ are functions of $I_L$ and $I_F$). We define $u^*_{excess}= u_{\text{excess}}|_{I_L=I_L^*\ \&\ I_F=I_F^*}$. Note that the bargaining would only occur if $u^*_{excess}>0$, i.e.  the framework creates additional joint profit that can be divided between the SPs.  Thus, henceforth, we characterize the NBS for the case that $u^*_{excess}>0$. We use $u_{excess}$ in the following theorem in which we characterize the NBS: 

\begin{theorem}\label{theorem:bargaining}\label{theorem:bargain_regulator}
\edit{	If  $u^*_{\text{excess}}> 0$, the optimum solution of the optimization \eqref{equ:nash_solution_2} is $(I_F^*,I_L^*,s^*)$ in which $I^*_L=I_{min,L}>0$, $I^*_F=I_{min,L}$ or $I^*_F=0$, and  $s^*$ is:
	\begin{equation}\label{equ:optimum_p}
	\begin{aligned}
	s^*&=\frac{1}{{I^*_F}^2}  \Big{(}n^*_F(p^*_F-c)-d_{F}-wu^*_{excess}\Big{)}
	\end{aligned}
	\end{equation}
	where $n^*_{F}=n_{F}|_{I_L=I^*_L,I_F=I^*_F}$, $p^*_F=p_F|_{I_L=I^*_L,I_F=I^*_F}$.
	}
\end{theorem}	
	\begin{remark}\label{remark:bargaining}
	The theorem characterizes $s^*$,$I^*_L$, and $I^*_F$ which directly lead to the NBS (using results of Theorem~\ref{theorem:stage 3}, and Equations \ref{equ:EUs_stage4}, \eqref{equ:payoffF} and \eqref{equ:payoffL}), i.e. $(\pi^*_{L},\pi^*_{F})$. To find the NBS before splitting the profit, SPs  cooperatively maximize the aggregate excess profit, $u_{excess}$ by solving the maximization problem ($\max_{I_F,I_L} u_{excess}$). Note that solution of this optimization is a collusive outcome in which both SPs decreases the amount of investment while maintaining a specific relative difference. Subsequently, they decide the split of the additional profit, i.e. the  payment paid to SP$_L$ by SP$_F$ ($s^* {I^*_F}^2$), based on \eqref{equ:optimum_p} which depends on the bargaining power each has ($w$ and $1-w$). 
\end{remark}
	\begin{proof}
	Note that	$d_{CP}$ and $d_{SP}$ are independent of $I_L$, $I_F$, and $s$. In addition, $u_{excess}$ is independent of $s$, and is only a function of $I_L$ and $I_F$. Thus, for a given $s$, it can be shown that (please refer to  equation (2) of \cite{bargainingsplit} for the proof), the optimum value of $s$ is such that:
	\begin{equation}\label{equ:optimum}
	\frac{\pi_{F}-d_{F}}{w}=\frac{\pi_L-d_{L}}{1-w}
	\nonumber
	\end{equation}
	if the solution for $s$ satisfies other constraints. Thus, by plugging the expressions of $\pi_F$ and $\pi_L$, \eqref{equ:payoffF} and \eqref{equ:payoffL}, in \eqref{equ:optimum}, and solving for $s$ yields the candidate optimum $s$:
	\begin{equation}\label{equ:optimum_p_2}
		\footnotesize
	\begin{aligned}
	s^*&=\frac{1}{{I_F}^2} \bigg{(}(1-w)\Big{(}n_F(p_F-c)-d_F\Big{)}-w\Big{(}n_L(p_L-c)-\gamma I_L^2-d_L\Big{)}\bigg{)}\\
	&=\frac{1}{{I_F}^2}  \Big{(}n_F(p_F-c)-d_{F}-wu_{excess}\Big{)}
	\end{aligned}
	\end{equation}
	\normalsize
	Substituting \eqref{equ:optimum_p_2} in the objective function of \eqref{equ:nash_solution_2} and using \eqref{equ:payoffF} and \eqref{equ:payoffL}  yield the new objective function:
	\begin{equation}
	w^w (1-w)^{1-w} u_{excess}
	\nonumber
	\end{equation}
	\normalsize
	Substituting \eqref{equ:optimum_p_2}, \eqref{equ:payoffF}, and \eqref{equ:payoffL}  in the constraint $\pi_{L}\geq d_{L}$ (similar to the previous step), yields the new constraint $(1-w)u_{excess}\geq 0$. Similar substitutions for $\pi_{F}\geq d_{F}$ yields $wu_{excess}\geq 0$. Thus, the optimization can be written as,
	\begin{equation}
	\small 
	\begin{aligned}
	&\max_{I_F,I_L} u_{excess}\\
	& \text{s.t.}\\
	&\qquad I_L>I_{min,L}\\
	&\qquad 0\leq I_F\leq I_L\\
	&\qquad u_{excess}\geq 0
	\end{aligned}
	\end{equation}
	\normalsize

Now, we rewrite the expression of $u_{\text{excess}}$ by using the definition of $u_{\text{excess}}$ \eqref{equ:uexcess}, the expressions for $p_L$ and $p_F$	in Theorem~\ref{theorem:stage 3}, and the expressions for $n_L$ and $n_F$ \eqref{equ:EUs_stage4}. First, by plugging the expressions of   $p_L$ and $p_F$	from Theorem~\ref{theorem:stage 3} in \eqref{equ:EUs_stage4}, we can find $n_L$ and $n_F$:
$$
\begin{aligned}
n_L=\frac{2I_L-I_F}{3I_L}\qquad \& \qquad n_F=\frac{I_L+I_F}{3I_L}
\end{aligned}
$$

Note that since $I_L\geq I_F$, $n_L>0$. In addition, $I_F>0$ yields $n_L<1$. Now, plugging these expressions and the expressions of $p_L$ and $p_F$ from Theorem~\ref{theorem:stage 3} in \eqref{equ:uexcess}, we have:
\begin{equation}
u_{\text{excess}}=\bigg{(}\frac{2I_L-I_F}{3I_L}\bigg{)}^2-\gamma I^2_L-d_{L}+\bigg{(}\frac{I_L+I_F}{3I_L}\bigg{)}^2-d_{F}
\end{equation} 
Consider $t=\frac{I_F}{I_L}$, then:
\begin{equation}\label{equ:uexcess_ilif}
u_{\text{excess}}=\bigg{(}\frac{2-t}{3}\bigg{)}^2-\gamma I^2_L-d_{L}+\bigg{(}\frac{1+t}{3}\bigg{)}^2-d_{F}
\end{equation} 

First, we argue that $I^*_L=I_{min,L}$. Suppose not. Then $I^*_L>I_{min,L}$. Consider $I^*_F$ to be the optimum solution of $I_F$ in this case. Take, $\hat{I}_L=I^*_{min,L}$ and $\hat{I}_F=I^*_F\frac{I_{min,L}}{I^*_L}$. Note that $t=\frac{I^*_F}{I^*_L}=\frac{\hat{I}_F}{I_{min,L}}$.  Thus, using \eqref{equ:uexcess_ilif}, since $t$  is constant and $I_L$ is lower than before, $u_{\text{excess}}$ would be higher with $\hat{I}_F$ and $I_{min,L}$ than with $I^*_F$ and $I^*_L$. This contradicts with $I^*_L$ and $I^*_F$ being optimum solutions. Thus, $I^*_L=I_{min,L}$. 

Now that we determined $I^*_L$, it is left to determine $t$. Note that the feasible set of $t$ is $0\leq t\leq 1$. From \eqref{equ:uexcess_ilif}, note that $\frac{d^2 u_{\text{excess}}}{d t^2}>0$.  Thus, $u_{excess}$  is convex, and the solution of the maximization of $\max_{0\leq t\leq 1} u_{excess}$  lie in the boundaries, i.e. $t^*=0$ or $t^*=1$. Plugging these values of $t$ in \eqref{equ:uexcess_ilif} yields the same payoff. Thus, $I^*_F=0$ or $I^*_F=I_{min,L}$.

These solution are the NBS only if they yield $u^*_{\text{excess}}>0$. If not, bargaining would not occur. This completes the proof. 

\end{proof}

\begin{remark} The price per sponsored resources \eqref{equ:optimum_p} is a decreasing function of $w$ , i.e. the bargaining power of SP$_F$: the higher the bargaining power of SP$_F$ the lower the  payment paid to SP$_L$. \footnote{Note that the framework allows for positive or negative payments. However, results reveal that for all the numerically tested parameters, this payment is positive.}
\end{remark}

\section{EUs with Outside Option}\label{section:generalization}
\edit{In the previous section, we considered a hotelling model with full market  coverage, i.e. when EUs should choose one of the two SPs. Thus, there is no outside option in the hotelling model. In this section, we consider the demand of SPs to be generated both through the hotelling model and also a demand function. The demand function captures two scenarios. It models (i) attrition in the number of EUs when the investment or price of  both SPs are not desirable enough, i.e. the effect of an outside option for EUs of hotelling model, and (ii) an exclusive additional customer base for each of the SPs to draw from depending on the investments and prices they offer.}

In Section~\ref{section:generalization_model}, we present the new model. In Section~\ref{section:generalziation_SPNE}, we characterize the SPNE strategies. We summarize the SPNE outcome in section~\ref{section:generalization_outcome}. In Section~\ref{section:generalization_bargaining}, we investigate the bargaining framework for this case. Finally, we provide numerical results in Section~\ref{section:generalization_simulation}.

\subsection{Model} \label{section:generalization_model}To model the outside option, we consider a demand function for each SP to be added to the demand from the hotelling model. Thus, 
\begin{definition}\label{definition:new_demand}
The fraction of EUs with each SP is
\begin{equation}\label{equ: demand}
\begin{aligned}
&\tilde{n}_{L}=n_{L}+\varphi_{L}(p_{L}, I_{L}),\\
&\tilde{n}_{F}=n_{F}+\varphi_{F}(p_{F}, I_{F}),
\end{aligned}
\end{equation}
where
\begin{equation}
\varphi_{L}(p_{L}, I_{L})=k-p_{L}+b(I_{L}-I_{F}),\footnote{Note that the cofficient of $p_L$ is normalizd to one.}\quad \varphi_{F}(p_{F}, I_{F})=k-p_{F}+bI_{F}
\end{equation}
and $k$ and $b$ are constants.
\end{definition}
Note that $\varphi_L(.,.)$ and $\varphi_F(.,.)$ can be positive or negative. A positive value denotes attracting EUs other than those considered in the hotelling model, and presumably from an exclusive additional customer base, and a negative value denotes losing some of the EUs considered in the hotelling model to an outside option, i.e. attrition in the number of EUs. Now, we proceed to characterize the SPNE using backward induction.
\subsection{SPNE Analysis}
\label{section:generalziation_SPNE}
We characterize the SPNE using backward induction starting from Stage 3:
\subsubsection*{Stage 3}
In this stage, $SP_{L}$ and $SP_{F}$ maximize their utilities by determining their access prices for EUs, $p_{L}$ and $p_{F}$, respectively. We mainly focus on interior Nash equilibria, i.e. $0<n_F<1$ and $0<n_L<1$. Using Definition~\ref{definition:new_demand}, \eqref{equ:t_L}, \eqref{equ:t_F}, and \eqref{equ:EUs_stage4} the payoffs of SPs are:
\begin{equation}\label{equ:stage 3}
\small
\begin{aligned}
\pi_{F}=&(t_{L}+k+p_{L}-2p_{F}+bI_{F})(p_{F}-c)-sI_{F}^{2}\\
\pi_{L}=&(t_{F}+k+p_{F}-2p_{L}+bI_{L}-bI_{F})(p_{L}-c)+sI_{F}^{2}-\gamma I_{L}^{2}
\end{aligned}
\end{equation}
In the following theorem, we characterize the unique interior NE:
\begin{theorem}\label{theory pl pf}
For given $I_{F}$ and $I_{L}$, the interior Nash strategies of stage 3 are unique, and are:
\begin{equation}\label{equ: price}
\begin{aligned}
p_{L}^{*}=&\frac{1}{15}+\frac{2c}{3}+\frac{k}{3}+\frac{t_{F}}{5}-\frac{b}{5}I_{F}+\frac{4b}{15}I_{L},\\
p_{F}^{*}=&\frac{1}{15}+\frac{2c}{3}+\frac{k}{3}+\frac{t_{L}}{5}+\frac{b}{15}I_{L}+\frac{b}{5}I_{F}.
\end{aligned}
\end{equation}
if and only if  $I_{L}$ satisfies:
\begin{equation}\label{equ: condition 1}
I_{L}<\frac{4}{b},
\end{equation}
\edit{If $I_L=\frac{4}{b}$, then $p^*_L$ and $p^*_F$ may constitute  corner or  interior equilibrium strategies.}
\end{theorem}

\begin{proof}
In this case, every  NE by which $0\leq x_n\leq 1$, should satisfy the first order condition. Thus $p_{L}^{*}$ and $p_{F}^{*}$ should be such that:
\begin{align*}
\frac{d\pi_{L}}{dp_{L}}|_{p_{L}^{*}}=0,\quad \frac{d\pi_{F}}{dp_{F}}|_{p_{F}^{*}}=0,
\end{align*}
then
\begin{equation}\label{equ: first order condition of prices}
\begin{aligned}
4p_{L}^{*}-p_{F}^{*}&=t_{F}+k+bI_{L}-bI_{F}+2c\\
 4p_{F}^{*}-p_{L}^{*}&=t_{L}+k+bI_{F}+2c.
\end{aligned}
\end{equation}
Add the two equations on both sides, note $t_{L}+t_{F}=1$, we have:
\begin{align}\label{equ: p_L+p_F}
p_{L}^{*}+p_{F}^{*}=\frac{1}{3}+\frac{2k}{3}+\frac{4c}{3}+\frac{b I_{L}}{3}
\end{align}
Substitute (\ref{equ: p_L+p_F}) into (\ref{equ: first order condition of prices}), then:
\begin{align*}
p_{L}^{*}=&\frac{1}{15}+\frac{2c}{3}+\frac{k}{3}+\frac{t_{F}}{5}-\frac{b}{5}I_{F}+\frac{4b}{15}I_{L},\\
p_{F}^{*}=&\frac{1}{15}+\frac{2c}{3}+\frac{k}{3}+\frac{t_{L}}{5}+\frac{b}{15}I_{L}+\frac{b}{5}I_{F}.
\end{align*}

Take the second derivative of $\pi_{L}$ with respect to $p_{L}$,
\begin{align*}
\frac{d^{2}\pi_{L}}{d(p_{L}^{*})^{2}}=\frac{d^{2}\pi_{F}}{d(p_{F}^{*})^{2}}=-4<0
\end{align*}
then $p_{L}^{*}$ and $p_{F}^{*}$ are the unique maximal solutions of $\pi_{L}$ and $\pi_{F}$, respectively.

Thus, $p^*_F$ and $p^*_L$ are the unique interior NE strategies if and only if $0<x_0<1$. Substitute (\ref{equ: price}), \eqref{equ:t_L}, and \eqref{equ:t_F} into (\ref{equ:xn}) yields:
\begin{align}\label{equ: x_0}
x_{0}=\frac{4}{5}-\frac{b}{5}I_{L}+(\frac{2b}{5}-\frac{3}{5I_{L}})I_{F}\triangleq \Psi(I_{F})
\end{align}
Once $I_L$ is fixed, $\Psi(I_{F})$ would be  a linear function of $I_{F}$. Thus, $0<\Psi(I_F)<1$  for any values of $I_F$ such that $0\leq I_F\leq I_L$, if and only if $0<\Psi(0)<1$ and $0<\Psi(I_L)<1$. Note that  $\Psi(I_{L})=\frac{1}{5}+\frac{b}{5}I_{L}\in(0, 1)$ when $0<I_{L}<\frac{4}{b}$, and $\Psi(0)=\frac{4}{5}-\frac{b}{5}I_{L}\in(0, 1)$ when $0<I_{L}<\frac{4}{b}$. Note that if $I_L=\frac{4}{b}$, then $x_0$ can close any value in $[0,1]$. Thus, $p^*_L$ and $p^*_F$ may constitute  corner or interior equilibria. This completes the proof of the theorem.
\end{proof}

To simplify the subsequent expressions, we define $f(I_L)$ and $g(I_L)$ as follows:
\begin{definition}
$g(I_{L})=\frac{b}{15}I_{L}+\frac{1}{15}-\frac{c}{3}+\frac{k}{3}$, $f(I_{L})=\frac{1}{5I_{L}}+\frac{b}{5}>0$
\end{definition}

\subsubsection*{Stage 2} 
In this stage, $\text{SP}_{F}$ decides on the investment level, with the condition that $0\leq I_{F}\leq I_{L}$, to maximize $\pi_{F}$. Substituting (\ref{equ: price}) into (\ref{equ:stage 3}), we have the following optimization problem:

\begin{lemma}
\begin{equation}\label{equ: stage 2}
\small
\begin{aligned}
\max_{I_{F}}\quad&\pi_{F}=(2f^{2}(I_{L})-s)I_{F}^{2}+4f(I_{L})g(I_{L})I_{F}+2g^{2}(I_{L})\\
s.t\quad &0\leq I_{F}\leq I_{L}\\
&\pi_{F}\geq0
\end{aligned}
\end{equation}
\end{lemma}

\begin{proof}
From (\ref{equ:stage 3}), we have $\pi_{F}=(t_{L}+k+p_{L}-2p_{F}+bI_{F})(p_{F}-c)-sI_{F}^{2}$. Note $t_{L}=I_{F}/I_{L}$ and $t_{F}=1-t_{L}$

\noindent{\textbf{(\romannumeral1)}} First, we substitute the values in $(t_{L}+k+p_{L}-2p_{F}+bI_{F})$:
\begin{equation}
\small
\begin{aligned}
&t_{L}+k+p_{L}-2p_{F}+bI_{F}=\\
&t_{L}+k-\frac{1}{15}-\frac{2c}{3}-\frac{k}{3}+\frac{t_{F}}{5}-\frac{2t_{L}}{5}+\frac{2b}{15}I_{L}-\frac{3b}{5}I_{F}+bI_{F}=\\
&\frac{3t_{L}}{5}+\frac{t_{F}}{5}+\frac{2k}{3}-\frac{1}{15}-\frac{2c}{3}+\frac{2b}{15}I_{L}+\frac{2b}{5}I_{F}=\\&\frac{2I_{F}}{5I_{L}}+\frac{2}{15}+\frac{2k}{3}-\frac{2c}{3}+\frac{2b}{15}I_{L}+\frac{2b}{5}I_{F}
=2g(I_{L})+2f(I_{L})I_{F}.
\end{aligned}
\end{equation}
\noindent{\textbf{(\romannumeral2)}} Then, we compute $p_{F}-c$,
\begin{equation}
\small
p_{F}-c=\frac{1}{15}-\frac{c}{3}+\frac{k}{3}+\frac{I_{F}}{5I_{L}}+\frac{b}{15}I_{L}+\frac{b}{5}I_{F}=g(I_{L})+f(I_{L})I_{F}
\end{equation}

\noindent{\textbf{(\romannumeral3)}} From (\romannumeral1) and (\romannumeral2), we have $\pi_{F}(p_{F}; p_{L}, s, I_{L}, I_{F})=2(g(I_{L})+f(I_{L})I_{F})^{2}-sI_{F}^{2}=(2f^{2}(I_{L})-s)I_{F}^{2}+4f(I_{L})g(I_{L})I_{F}+2g^{2}(I_{L})$.

This completes the proof.
\end{proof}

\begin{theorem}\label{theory I_F}
For any given $s$, $I_{L}$, let $I_{F}^{0}$ to be the solution to the first order condition on $\pi_F$:
\begin{equation}
I_{F}^{0}=\frac{-2f(I_{L})g(I_{L})}{2f^{2}(I_{L})-s}
\end{equation}
Then, the unique optimal investment level of $\text{SP}_{F}$, $I_{F}^{*}$, is:
{\begin{equation}\label{equ: I_F^*}
\footnotesize
\begin{aligned}
&I_{F}^{*}=I_{F}^{0} \,\,\text{if}\quad (s, I_{L})\in\left\{s>2f^{2}(I_{L})+2f(I_{L})g(I_{L})/I_{L},\, g(I_{L})\geq0\right\}\\
&I_{F}^{*}=I_{L} \,\,\text{if}\quad
(s, I_{L})\in\{2f^{2}(I_{L})\leq s\leq2f^{2}(I_{L})+2f(I_{L})g(I_{L})/I_{L},\,\\
& g(I_{L})\geq0\}\cup\{2f^{2}(I_{L})+4f(I_{L})g(I_{L})/I_{L}\geq s,\, 2f^{2}(I_{L})>s\}
\end{aligned}
\end{equation}}
{, if $\pi_F(I^*_F;I_L)\geq 0$. If  $\pi_F(I^*_F;I_L)< 0$, then no cooperation occurs between the MNO and the MVNO.}
\end{theorem}

\begin{proof}
We consider different cases. First, we consider the case that $2f^{2}(I_{L})-s=0$ (Step (i)). Then, we consider the case that $2f^{2}(I_{L})-s\neq0$ and $\pi_F$ is a quadratic function of $I_F$ (Step (ii)). In Step (iii), we prove that $I^*_F\neq 0$. Combining the steps yields the result of the theorem.

\noindent{\textbf{Step (\romannumeral1)}:} If $2f^{2}(I_{L})-s=0$, $\pi_{F}$ is linear function of $I_F$, i.e.,
\begin{align*}
\pi_{F}=4f(I_{L})g(I_{L})I_{F}+2g^{2}(I_{L})
\end{align*}
Thus,
\begin{align*}
\left\{
\begin{aligned}
&I_{F}^{*}=0& \,\,&\text{if}\quad g(I_{L})<0\\
&I_{F}^{*}=I_{L}& \,\,&\text{if}\quad g(I_{L})\geq0
\end{aligned}
\right..
\end{align*}

\noindent{\textbf{Step (\romannumeral2)}:} Now, consider the case that $2f^{2}(I_{L})-s\neq0$ and $\pi_F$ is a quadratic function of $I_F$. We characterize the optimum answer in two cases: (a) if $2f^{2}(I_{L})-s>0$, and (b) if $2f^{2}(I_{L})-s<0$, $\pi_{F}(I_{F}; I_{L})$.

For the case that $\pi_{F}$ is a quadratic function, we use   the solution to the first order condition ($I_{F}^{0}$):  
\begin{align*}
\frac{d\pi_{F}}{dI_{F}}|_{I_{F}^{0}}=0\Rightarrow I_{F}^{0}=\frac{-2f(I_{L})g(I_{L})}{2f^{2}(I_{L})-s}
\end{align*}

Note that the symmetric axis of the quadratic function is $y=I_{F}^{0}$. This is the vertical line that divide the quadratic function into two identical halves. 

\noindent{\textbf{Case (ii-a)}}: If $2f^{2}(I_{L})-s>0$, then $\pi_{F}$ is convex function,. Thus, the maximal value must be obtained at the boundaries, i.e., $0$ or $I_{L}$. To find the maximum value, we compare the symmetric axis $y=I_{F}^{0}$ with  the midpoint of the two boundary points \footnote{Note that the function is symmetric around the symmetric axis. Thus, if the midpoint is at the right (respectively, left) of the symmetric axis, then the right (respectively, left) boundary yields the highest value.}. The midpoint of boundary points is $I_{L}/2$. Hence,
\begin{align*}
\left\{
\begin{aligned}
&I_{F}^{0}-\frac{I_{L}}{2}\leq0& \,\,&\text{if}\quad 2I_{L}f^{2}(I_{L})+4f(I_{L})g(I_{L})-I_{L}s\geq0\\
&I_{F}^{0}-\frac{I_{L}}{2}>0& \,\,&\text{if}\quad 2I_{L}f^{2}(I_{L})+4f(I_{L})g(I_{L})-I_{L}s<0
\end{aligned}
\right.,
\end{align*}
thus
\begin{align*}
\left\{
\begin{aligned}
&I_{F}^{*}=I_{L}& \,\,&\text{if}\quad 2I_{L}f^{2}(I_{L})+4f(I_{L})g(I_{L})-I_{L}s\geq0\\
&I_{F}^{*}=0& \,\,&\text{if}\quad 2I_{L}f^{2}(I_{L})+4f(I_{L})g(I_{L})-I_{L}s<0
\end{aligned}
\right..
\end{align*}

\noindent{\textbf{Case (ii-b}):} If $2f^{2}(I_{L})-s<0$, then $\pi_{F}$ is a concave function. Thus, the maximal value is either at $0$, $I_{L}$ or $I_{F}^{0}$. Now we  compare the symmetric axis $y=I_{F}^{0}$ with the boundary points to find the maximum value:
\begin{align*}
\footnotesize
\left\{
\begin{aligned}
&I_{F}^{0}-0<0& \,\,&\text{if}\quad g(I_{L})<0\\
&0\leq I_{F}^{0}< I_{L}& \,\,&\text{if}\quad 2I_{L}f^{2}(I_{L})+2f(I_{L})g(I_{L})-I_{L}s<0,\quad g(I_{L})\geq0\\
&I_{F}^{0}\geq I_{L}& \,\,&\text{if}\quad
2I_{L}f^{2}(I_{L})+2f(I_{L})g(I_{L})-I_{L}s\geq0,\quad g(I_{L})\geq0
\end{aligned}
\right.,
\end{align*}
Thus,
\begin{align*}
\footnotesize
\left\{
\begin{aligned}
&I_{F}^{*}=0& \,\,&\text{if}\quad g(I_{L})<0\\
&I_{F}^{*}=I_{F}^{0}& \,\,&\text{if}\quad2I_{L}f^{2}(I_{L})+2f(I_{L})g(I_{L})-I_{L}s<0,\quad g(I_{L})\geq0\\
&I_{F}^{*}=I_{L}& \,\,&\text{if}\quad
2I_{L}f^{2}(I_{L})+2f(I_{L})g(I_{L})-I_{L}s\geq0,\quad g(I_{L})\geq0
\end{aligned}
\right..
\end{align*}

\noindent{\textbf{Step (\romannumeral3)}:} We now prove $I_{F}^{*}\neq0$. From Case (ii-a), if $I_{F}^{*}=0$, then $2I_{L}f^{2}(I_{L})+4f(I_{L})g(I_{L})-I_{L}s<0$, i.e., $s>2f^{2}(I_{L})+4f(I_{L})g(I_{L})/I_{L}$, which implies $g(I_{L})<0$ since $2f^{2}(I_{L})-s>0$. Thus from Step (i), and Cases (ii-a) and (ii-b), if $I_{F}^{*}=0$,  then $g(I_{L})<0$.

Since $t_{L}^{*}=0$ and $t_{F}^{*}=1$, when $I_{F}^{*}=0$, then $p_{F}^{*}-c=\frac{1}{15}-\frac{c}{3}+\frac{k}{3}+\frac{b}{5}I_{F}=g(I_{L})<0$. For an equilibrium solution $p_{F}^{*}$, $p_{F}^{*}\geq c$, otherwise $\pi_{L}^{*}=n^{*}_{F}(p_{F}^{*}-c)-s(I_{F}^{*})^{2}<0$. Hence $I_{F}^{*}=0$ can not be an equilibrium solution for $\text{SP}_{F}$.

Combining Steps (\romannumeral1), (\romannumeral2),  and (\romannumeral3), yields:
{\begin{equation}\label{equ: I_F^*}
\footnotesize
\begin{aligned}
&I_{F}^{*}=I_{F}^{0} \,\,\text{if}\quad (s, I_{L})\in\left\{s>2f^{2}(I_{L})+2f(I_{L})g(I_{L})/I_{L},\, g(I_{L})\geq0\right\}\\
&I_{F}^{*}=I_{L} \,\,\text{if}\quad
(s, I_{L})\in\{2f^{2}(I_{L})\leq s\leq2f^{2}(I_{L})+2f(I_{L})g(I_{L})/I_{L},\,\\
& g(I_{L})\geq0\}\cup\{2f^{2}(I_{L})+4f(I_{L})g(I_{L})/I_{L}\geq s,\, 2f^{2}(I_{L})>s\}
\end{aligned}
\end{equation}}
Intuitively, $I^*_F$ is the optimum solution  if it yields  $\pi_F\geq 0$. If not, i.e. $\pi_F< 0$, then no cooperation occurs between the MNO and the MVNO. The theorem follows.
\end{proof}

\subsubsection*{Stage 1} In this stage, MNO decides on the level of investment $I_{L}$  to maximize his payoff $\pi_{L}$:
\begin{theorem}\label{theorem:generalization_stage1} The optimum investment level of the MNO is the solution to the following optimization problem:
\begin{equation}\label{equ: stage 1}
\footnotesize
\begin{aligned}
\max_{I_{L}}\quad&\pi_{L}(I_{L}; I_{F}^{*})=2(\frac{b}{5}I_{L}+\frac{1}{5}+g(I_{L})-f(I_{L})I_{F}^{*})^{2}+s(I_{F}^{*})^{2}-\gamma I_{L}^{2}\\
s.t\quad&0<I_{L}\leq\frac{4}{b}\\
&\pi_{L}(I_{L})\geq0
\end{aligned}
\end{equation}
\end{theorem}

\begin{proof}
From (\ref{equ:stage 3}), we have $\pi_{L}=(t_{F}+k+p_{F}-2p_{L}+bI_{L}-bI_{F})(p_{L}-c)+sI_{F}^{2}-\gamma I_{L}^{2}$.

\noindent{\textbf{(\romannumeral1)}} First, we substitute the values in $t_{F}+k+p_{F}-2p_{L}+bI_{L}-bI_{F}$. We use the expressions for $t_L$ and $t_F$, i.e. $t_{L}=I_{F}/I_{L}$ and $t_{F}=1-t_{L}$:
\begin{align*}
\footnotesize
&t_{F}+k+p_{F}-2p_{L}+bI_{L}-bI_{F}=t_{F}+k-\frac{1}{15}-\frac{2c}{3}\\
&-\frac{k}{3}+\frac{t_{L}}{5}-\frac{2t_{F}}{5}-\frac{7b}{15}I_{L}+\frac{3b}{5}I_{F}+bI_{L}-bI_{F}\\
=&\frac{3t_{F}}{5}+\frac{2k}{3}-\frac{1}{15}-\frac{2c}{3}+\frac{t_{L}}{5}+\frac{8b}{15}I_{L}-\frac{2b}{5}I_{F}\\
=&\frac{2t_{F}}{5}+\frac{2k}{3}+\frac{2}{15}-\frac{2c}{3}+\frac{8b}{15}I_{L}-\frac{2b}{5}I_{F}\\
=&2g(I_{L})+\frac{2}{5}-\frac{2I_{F}}{5I_{L}}-\frac{2b}{5}I_{F}+\frac{2b}{5}I_{L}\\
=&2(\frac{b}{5}I_{L}+\frac{1}{5}+g(I_{L})-f(I_{L})I_{F}).
\end{align*}
\noindent{\textbf{(\romannumeral2)}} Then, we compute $p_{L}-c$:
\begin{align*}
\footnotesize
p_{L}-c&=\frac{1}{15}-\frac{c}{3}+\frac{k}{3}+\frac{1}{5}-\frac{I_{F}}{5I_{L}}-\frac{b}{5}I_{F}+\frac{4b}{15}I_{L}\\
&=\frac{b}{5}I_{L}+\frac{1}{5}+g(I_{L})-f(I_{L})I_{F}.
\end{align*}

\noindent{\textbf{(\romannumeral3)}} From (\romannumeral1) and (\romannumeral2), we have $\pi_{L}={2(\frac{b}{5}I_{L}+\frac{1}{5}+g(I_{L})-f(I_{L})I_{F}^{*})^{2}+s(I_{F}^{*})^{2}-\gamma I_{L}^{2}}$. The theorem follows.
\end{proof}

Later, in Section~\ref{section:generalization_simulation}, we find the solution to the optimization in Theorem~\ref{theorem:generalization_stage1}, using numerical analysis.

\subsection{The SPNE Outcome}\label{section:generalization_outcome}
We characterize the equilibrium outcome of the game using the results of previous analysis.

\begin{corollary}
The equilibrium outcome of the game is
\begin{itemize}
  \item Stage 1: $I_{L}^{*}$ is characterized by solving the optimization in Theorem~\ref{theorem:generalization_stage1}. 
  \item Stage 2:  $I_{F}^{*}$ is characterized in Theorem~\ref{theory I_F}. 
  \item Stage 3:
$p_{L}^{*}=\frac{1}{15}+\frac{2c}{3}+\frac{k}{3}+\frac{I_{L}^{*}-I_{F}^{*}}{5I_{L}^{*}}-\frac{b}{5}I_{F}^{*}+\frac{4b}{15}I_{L}^{*},\quad p_{F}^{*}=\frac{1}{15}+\frac{2c}{3}+\frac{k}{3}+\frac{I_{F}^{*}}{5I_{L}^{*}}+\frac{b}{15}I_{L}^{*}+\frac{b}{5}I_{F}^{*}$.

  \item Stage 4: $\tilde{n}_{L}^{*}=\frac{I_{L}^{*}-I_{F}^{*}}{I_{L}^{*}}+p_{F}^{*}-2p_{L}^{*}+k+bI_{L}^{*}-bI_{F}^{*},\quad \tilde{n}_{F}^{*}=\frac{I_{F}^{*}}{I_{L}^{*}}+p_{L}^{*}-2p_{F}^{*}+k+bI_{F}^{*}$.
\end{itemize}
\end{corollary}
\begin{proof}
Stage 1 and 2 are immediate from Theorems~\ref{theory I_F} and \ref{theorem:generalization_stage1}. The results of Stage 3 is found by substituting the results of Theorem~\ref{theory I_F} in the expressions for $p^*_F$ and $p^*_L$ in Theorem~\ref{theory pl pf}. The result of Stage 4, is found by substituting the results of previous stages in  \eqref{equ: demand}.
\end{proof}

\subsection{Bargaining Framework}\label{section:generalization_bargaining}
In Section~\ref{section:bargaining}, we investigated the bargaining framework for the hotelling model. In Theorem~\ref{theorem:bargaining}, we characterized the NBS of the problem. Since we do not used the explicit expressions of the parameters, we can use the results of this theorem. Thus, to find the NBS, we should solve the following optimization problem: 
\begin{align*}
\max_{I_{F}, I_{L}}\quad&u_{excess}\\
s.t.\quad&0<I_{L}\leq \frac{4}{b}\\
&0\leq I_{F}\leq I_{L}\\
&u_{excess}> 0
\end{align*}
, where the first inequality is from Theorem \ref{theory pl pf} and is essential for having an interior NE strategy. 

First, we find the expression for $u_{excess}$ in the following lemma:
\begin{lemma}
For any $d_{F}$ and $d_{L}$, $u_{excess}=4f^{2}(I_{L})I_{F}^{2}-4f^{2}(I_{L})I_{L}I_{F}+2g^{2}(I_{L})+2(f(I_{L})I_{L}+g(I_{L}))^{2}-\gamma I_{L}^{2}-d_{F}-d_{L}$.
\end{lemma}

\begin{proof}
Since $u_{excess}=\pi_{L}-d_{L}+\pi_{F}-d_{F}$, then from (\ref{equ: stage 2}) and (\ref{equ: stage 1}), 
\begin{equation}
\footnotesize
\begin{aligned}
&u_{excess}=\pi_{L}-d_{L}+\pi_{F}-d_{F}=(2f^{2}(I_{L})-s)I_{F}^{2}\\&+4f(I_{L})g(I_{L})I_{F}+2g^{2}(I_{L})+2(\frac{b}{5}I_{L}+\frac{1}{5}+g(I_{L})-f(I_{L})I_{F})^{2}\\
&+sI_{F}^{2}-\gamma I_{L}^{2}-d_{F}-d_{L}\\
\quad =&2f^{2}(I_{L})I_{F}^{2}+4f(I_{L})g(I_{L})I_{F}+2g^{2}(I_{L})-4(\frac{b}{5}I_{L}+\frac{1}{5}\\
&+g(I_{L}))f(I_{L})I_{F}+2f^{2}(I_{L})I_{F}^{2}+2(\frac{b}{5}I_{L}+\frac{1}{5}+g(I_{L}))^{2}-\gamma I_{L}^{2}\\
&-d_{F}-d_{L}\\
=&4f^{2}(I_{L})I_{F}^{2}+2g^{2}(I_{L})-4(\frac{b}{5}I_{L}+\frac{1}{5})f(I_{L})I_{F}\\
&\qquad+2(\frac{b}{5}I_{L}+\frac{1}{5}+g(I_{L}))^{2} -\gamma I_{L}^{2}-d_{F}-d_{L}\\
=&4f^{2}(I_{L})I_{F}^{2}-4f^{2}(I_{L})I_{L}I_{F}+2g^{2}(I_{L})+2(f(I_{L})I_{L}+g(I_{L}))^{2}\\
&\qquad -\gamma I_{L}^{2} -d_{F}-d_{L}
\end{aligned}
\nonumber
\end{equation}
\end{proof}

In the next Theorem, we prove that in the NBS, the MVNO either reserves all the available resources or zero resources from the MNO:
\begin{theorem}\label{theorem:g_bargain}
 If $u_{excess}^{*}>0$, in the NBS $I_{F}^{*}=I_{L}^{*}$ or $I_{F}^{*}=0$. 
\end{theorem}
\begin{remark}
In this case, NBS is a milder version of the collusive outcome of the previous case in which EUs have no outside option.  Note that SP$_L$ no longer opts for the minimum possible $I_L$. However, in the NBS, SP$_F$ either reserves all the resources ($I^*_F=I^*_L$) or no resources ($I^*_F=0$).  
\end{remark}
\begin{proof}
The second derivative of $u_{excess}$ with respect to $I_F$ is:
$$
\frac{d^2 u_{excess}}{d I^2_F}=8 f^2(I_L)>0
$$
Thus, $u_{excess}$ is convex with respect to $I_F$, and the maximum of $u_{excess}$ is obtained at the boundaries of $I_F$, i.e. $I^*_F=0$ or $I^*_F=I^*_L$. This completes the proof of this theorem.
\end{proof}

\subsection{Numerical Results}
\label{section:generalization_simulation}
In this section, we provide numerical results. First, we focus on the sequential framework and the SPNE outcome characterized in Section~\ref{section:generalization_outcome}. Then, we provide the numerical results for ther bargaining framework. 

\subsubsection{SPNE Outcome}

\begin{figure}[t]
	\begin{subfigure}{.25\textwidth}
		\centering
		\includegraphics[width=\linewidth]{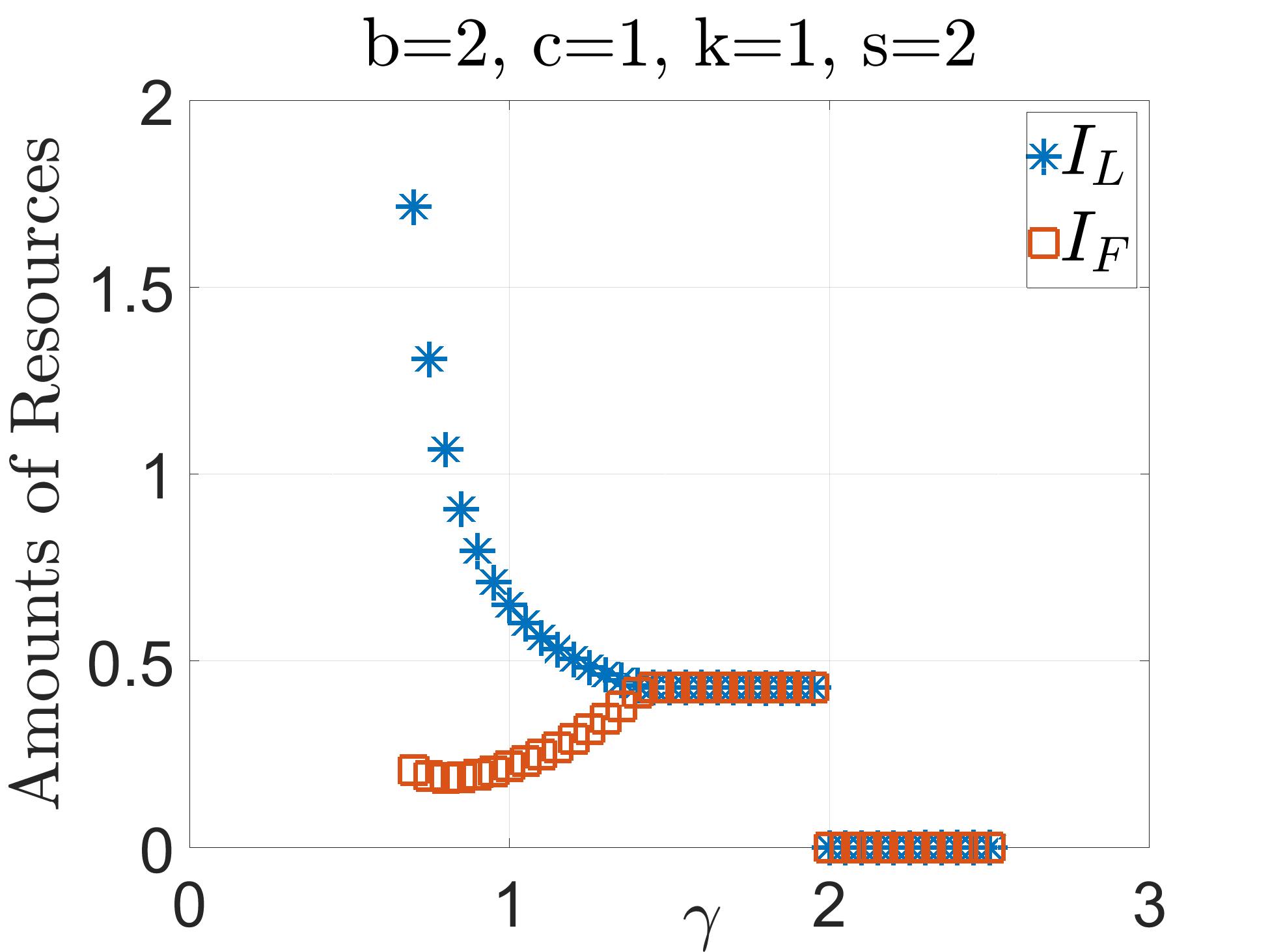}
		\label{figure:g_resources_gamma}
	\end{subfigure}%
	\begin{subfigure}{.25\textwidth}
		\centering
		\includegraphics[width=\linewidth]{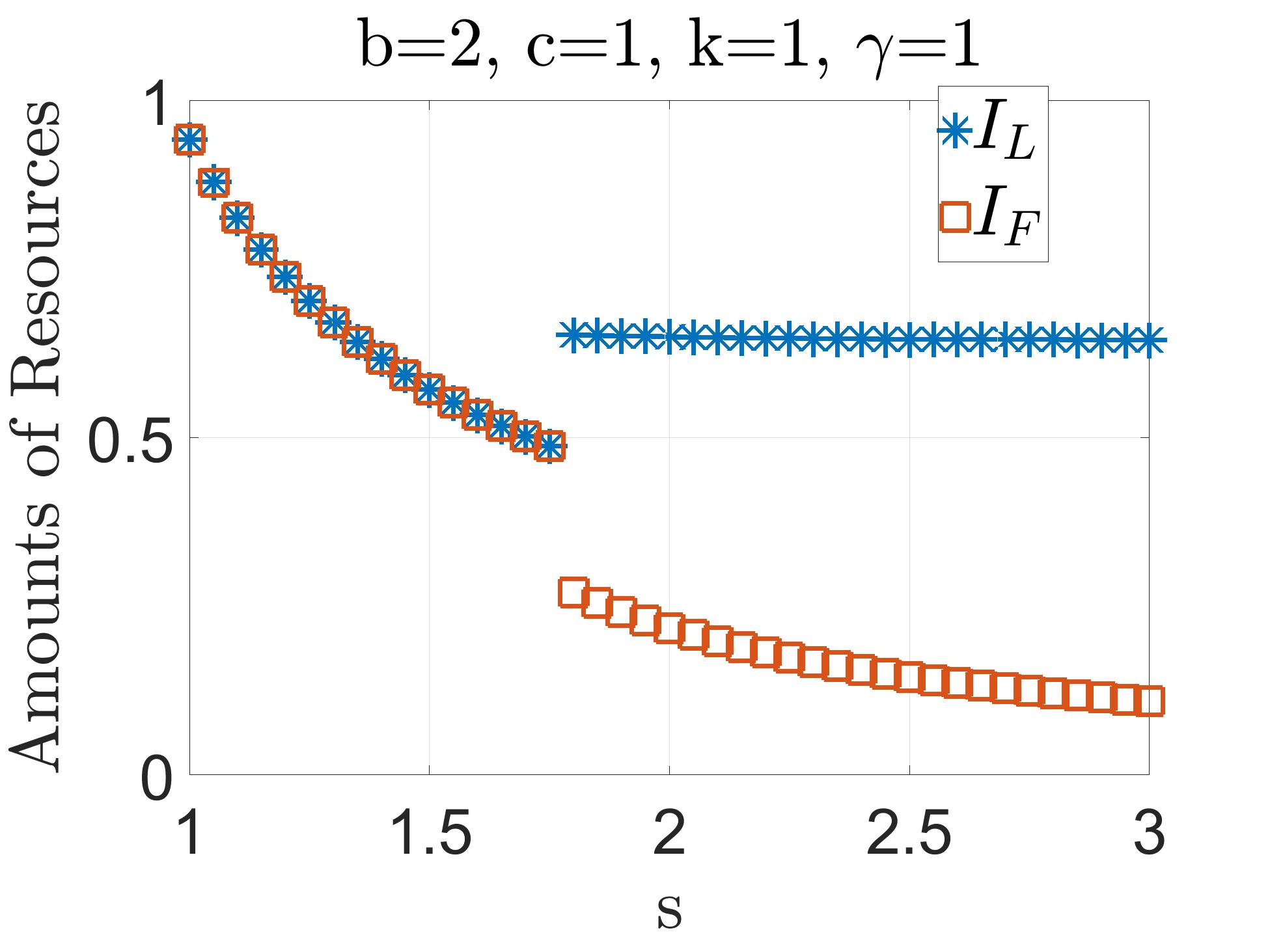}
		\label{fig:g_resources_s}
	\end{subfigure}
	\caption{Investment decisions of SPs vs. $\gamma$(left) and $s$ (right)}\label{figure:g_resources}
\end{figure}

\begin{figure}[t]
	\begin{subfigure}{.25\textwidth}
		\centering
		\includegraphics[width=\linewidth]{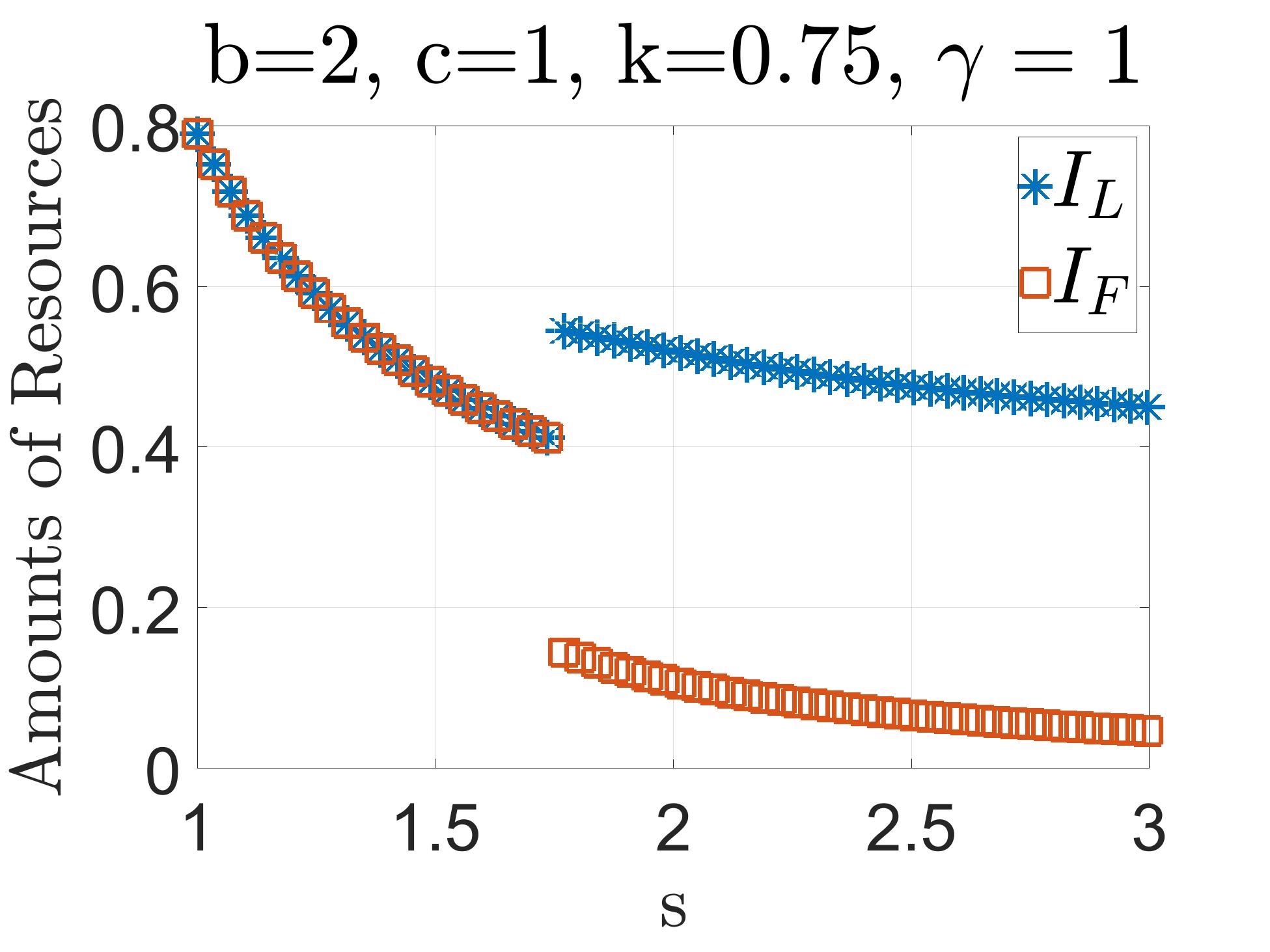}
		\label{figure:g_resources_gamma_k}
	\end{subfigure}%
	\begin{subfigure}{.25\textwidth}
		\centering
		\includegraphics[width=\linewidth]{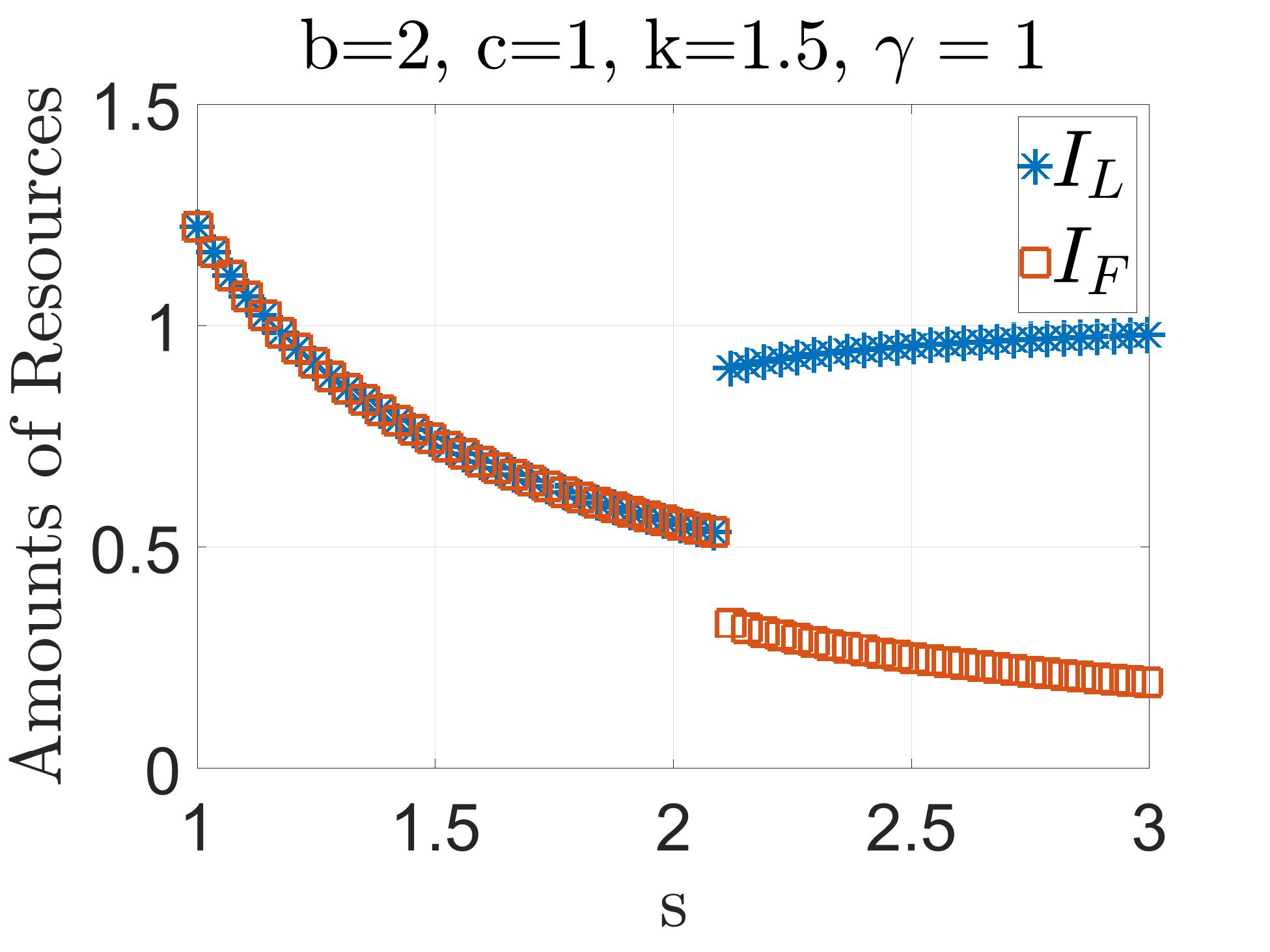}
		\label{fig:g_resources_s_k}
	\end{subfigure}
	\caption{Investment decisions of SPs vs. $s$  for $k=0.75$(left) and $k=1.5$ (right)}\label{figure:g_resources_k}
\end{figure}

In Figure~\ref{figure:g_resources}-left, we plot the optimum investment level of SPs with respect to marginal cost  of investment($\gamma$). Note that the number of resources invested by SP$_L$ is decreasing with $\gamma$, and the number of resources that SP$_F$ leases from SP$_L$ is increasing with $\gamma$ (Figure~\ref{figure:resources}-left). The reason for the former is intuitive: the higher the cost, the lower the investment. The latter happens since increasing $\gamma$ while $s$ is fixed increases the desirability of  the resources for SP$_F$. Note that in model (Section~\ref{section:model}), we assumed that to have a non-trivial problem the per resource fee should be higher than the marginal cost of investment, i.e. $s\geq \gamma$. This is the reason that in Figure~\ref{figure:g_resources}-left, when $\gamma> s=2$ the investment by the SPs is zero.

In Figure~\ref{figure:g_resources}-right, we plot the optimum investment level of SPs with respect to the per resource fee that SP$_F$ pays to SP$_L$. Note that the overall behavior of the number of resources with respect to $s$ is similar to that of the previous model presented in Figure~\ref{figure:resources}. The only difference is the behavior of the investment by SP$_L$ ($I_L$) when $s$ is higher than a threshold. In previous results (Figure~\ref{figure:resources}), this value was decreasing with $s$. However, in this model, the behavior of $I_L$ with respect to $s$ depends on the parameters of the demand function ($\varphi(.,.)$). We investigate this dependency by considering different values of $k$, i.e. the bias constant in the demand function. Results reveal that when $k=1$ (Figure~\ref{figure:g_resources}-right), then $I_L$ is constant with $s$, when $k=0.75$ (Figure~\ref{figure:g_resources_k}-left), $I_L$ is decreasing, and when $k=1.5$ (Figure~\ref{figure:g_resources_k}-right), $I_L$ is increasing with $s$.  Note that the higher $k$, the higher the number of EUs generated with the demand function $\varphi(.,.)$, and the more important this demand to the SPs. In addition, the demand function is increasing with the investment. Thus, as $k$ increases, SP$_L$ has more incentive to invest more to attract a higher number of EUs. 

\begin{figure}[t]
	\begin{subfigure}{.25\textwidth}
		\centering
		\includegraphics[width=\linewidth]{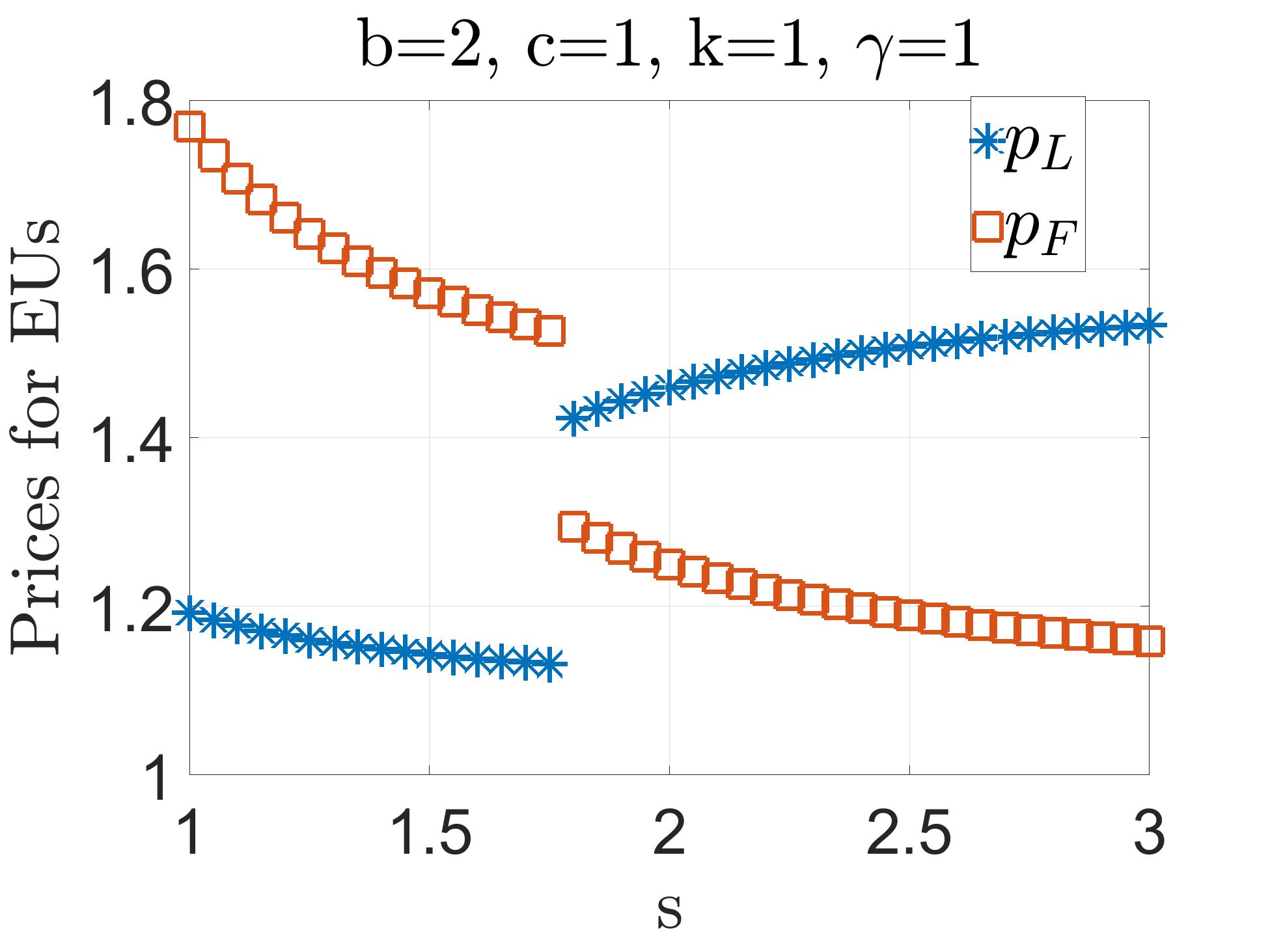}
		\label{figure:g_p_s}
	\end{subfigure}%
	\begin{subfigure}{.25\textwidth}
		\centering
		\includegraphics[width=\linewidth]{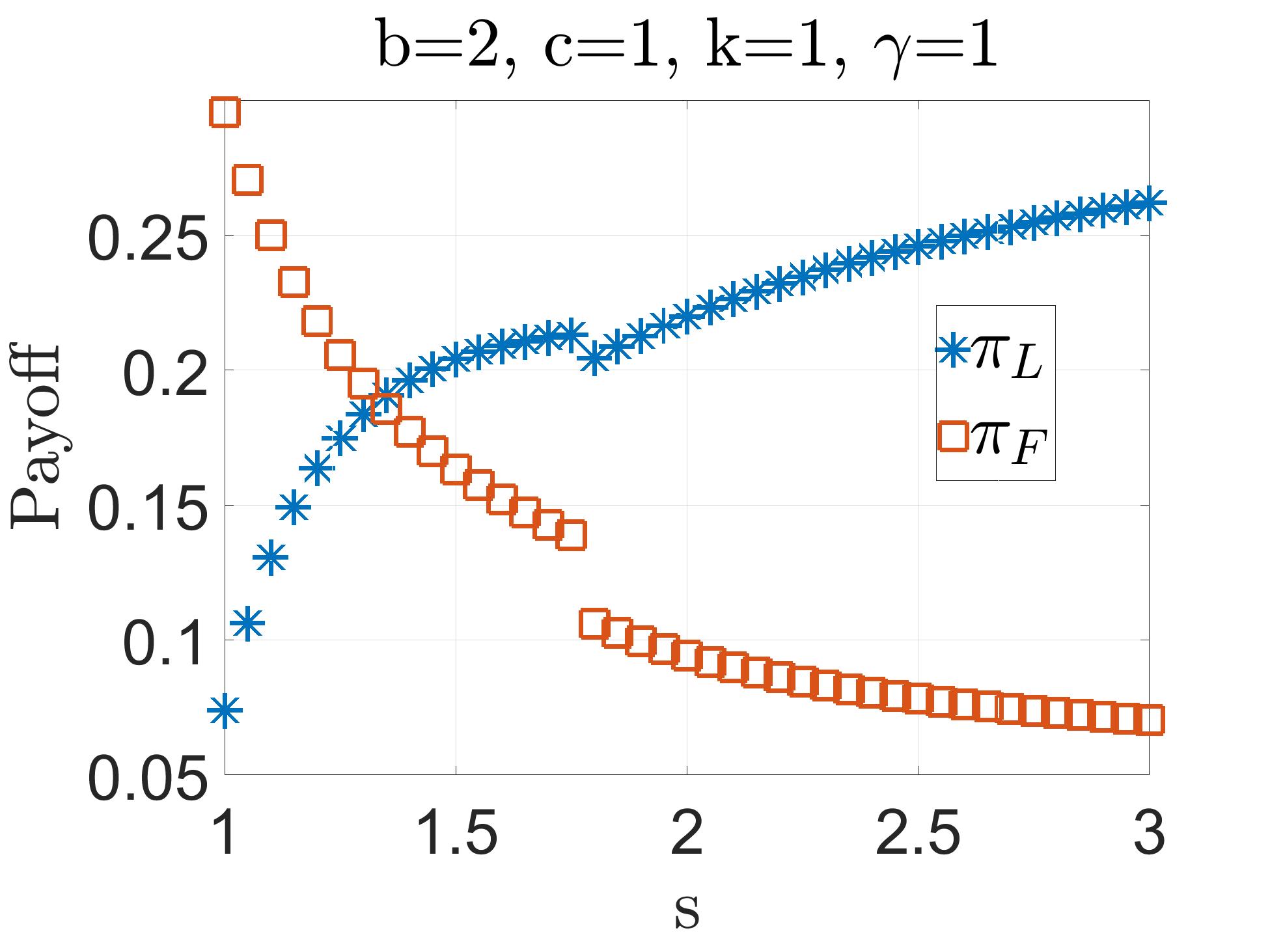}
		\label{fig:g_pn_s}
	\end{subfigure}
	\caption{Prices for EUs (left) and payoff of SPs (right) vs. $s$}\label{figure:gg-resources}
\end{figure}

In Figure~\ref{figure:gg-resources}-left, we plot the pricing decision of SPs for EUs, i.e. $p^*_L$ and $p^*_F$. Note that  results share many similarities with those of the previous model (Figure~\ref{figure:prices}). The main difference is that in the region associated with small $s$'s, the prices that SPs quote for EUs are no longer constant and independent of $s$. The reason is because of considering the effect of a third option and having a demand function that does not react only to the differences in the investment  of the SPs (relative difference of SPs), but also to the absolute value of the investments. More specifically, note that in the previous model, for small $s$'s, SP$_F$ reserves all the new resources of SP$_L$ (Figure~\ref{figure:resources}). Thus, by the definition of the $t_j$'s \eqref{equ:t_L} and \eqref{equ:t_F}, the decision of EUs in the hotelling model  would be independent of the level of investments and subsequently independent of $s$. However, in the new model, although again for small $s$'s, SP$_F$ reserves all the new resources of SP$_L$, the decision of EUs associated  with the newly added demand functions ($\varphi(.,.)$'s) would be still dependent on the resources and $s$. This yields that the prices for EUs are no longer constant when $s$ is smaller than a threshold. Results also reveal that similar to the previous model, the behavior of the number of EUs with SPs ($\tilde{n}_L$ and $\tilde{n}_F$) with respect to $s$ follows the same trend as the behavior of the prices.    

In Figure~\ref{figure:gg-resources}-right, we plot the payoff of SPs with respect to the per resource fee, $s$. Note that the payoff of SP$_L$ is increasing and the payoff of SP$_F$ is decreasing with respect to $s$. 

\subsubsection{Bargaining Framework} 
\begin{figure}[t]
	\begin{subfigure}{.25\textwidth}
		\centering
		\includegraphics[width=\linewidth]{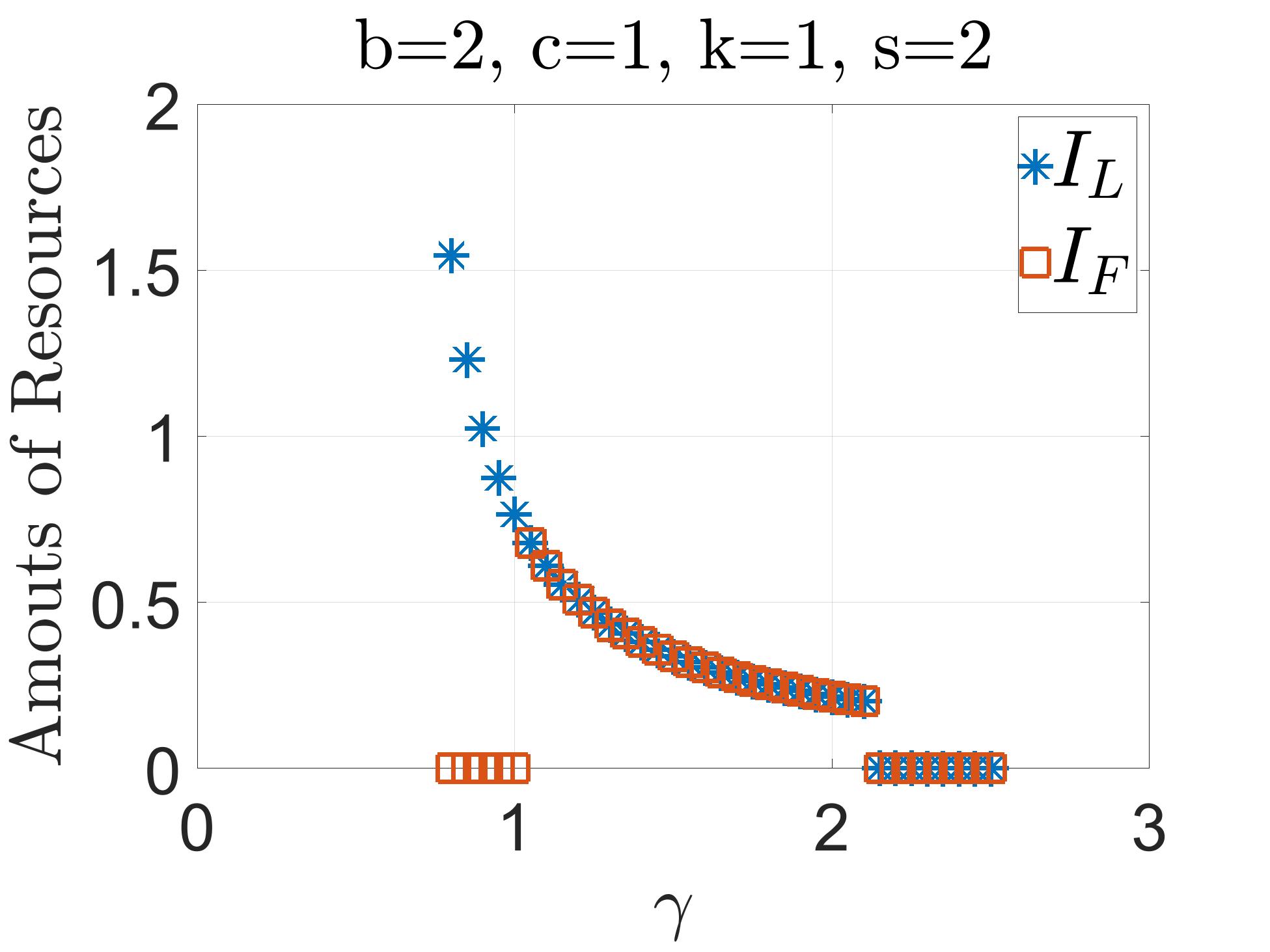}
		\label{figure:g_bar_resources}
	\end{subfigure}%
	\begin{subfigure}{.25\textwidth}
		\centering
		\includegraphics[width=\linewidth]{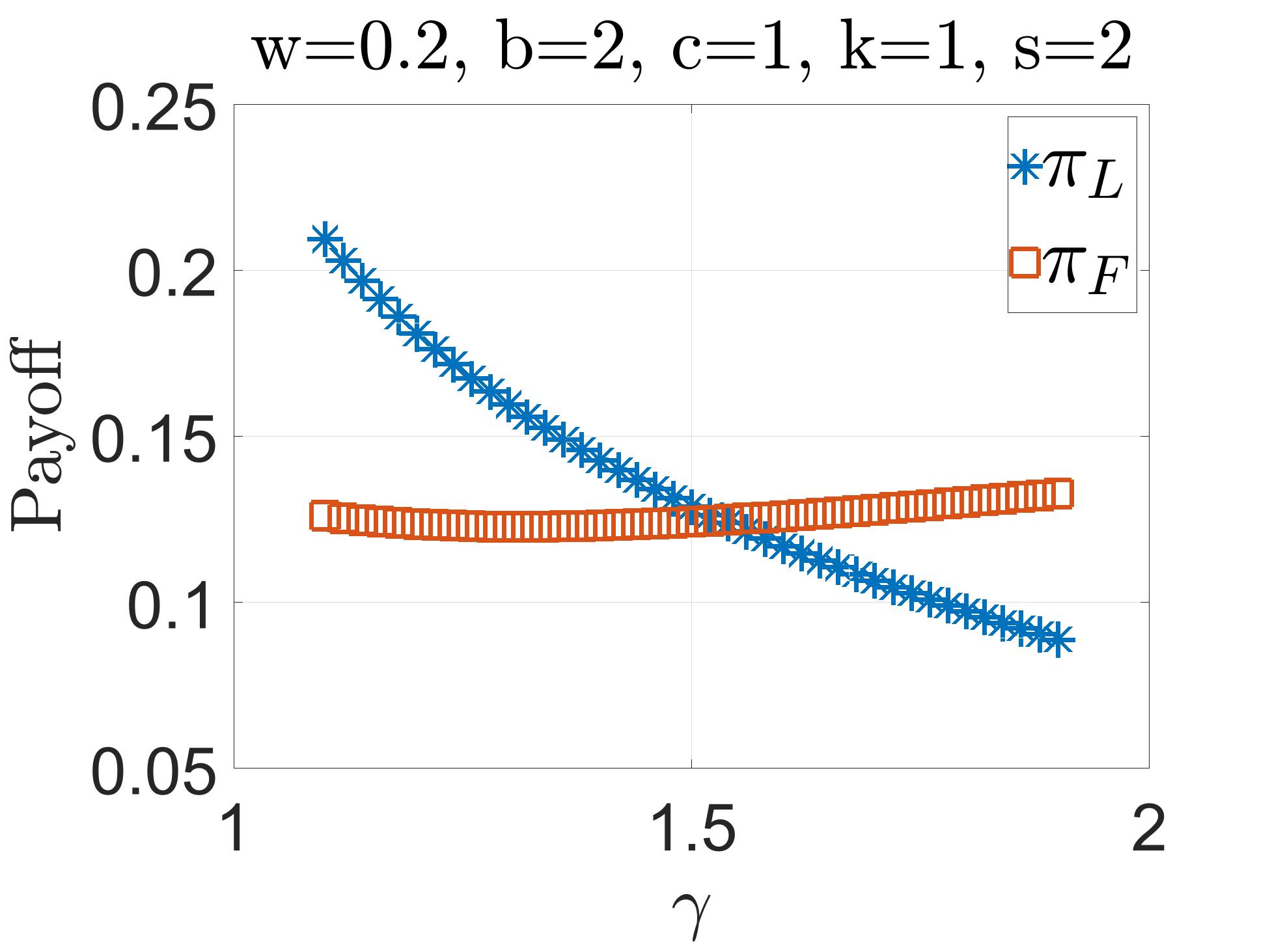}
		\label{fig:g_bar_pi}
	\end{subfigure}
	\caption{Investment decisions of SPs (left) and payoffs (right) in the bargaining framework vs. $\gamma$}\label{figure:g-bar}
\end{figure}
Now, we focus on the bargaining framework. Note that in Theorem~\ref{theorem:g_bargain}, we proved that in the bargaining framework, SP$_F$ either reserves all the newly created resources or no resources from SP$_L$. Numerical results also confirm this (Figure~\ref{figure:g-bar}-left). Recall that  in the previous model (Theorem~\ref{theorem:bargain_regulator}), in a bargaining framework, SP$_L$ reserves a minimum number of resources. However, in the  new model, $I^*_L$ is a solution of an optimization problem and does not necessarily assume values on the boundaries, i.e. minimum value (Figure~\ref{figure:g-bar}-left). 

In Figure~\ref{figure:g-bar}-right, we plot the payoff of SPs with respect to the marginal cost of investment ($\gamma$) after the bargaining when the relative bargaining power of SP$_F$ over SP$_L$  ($w$) is 0.2. Results reveal that when the marginal cost of investment is smaller (respectively, larger) than a threshold, then SP$_L$ (respectively, SP$_F$) is the SP that receives a higher payoff. 


\section{Conclusion}\label{section:conclusion}
We investigated the incentives of a Mobile Network Operators (MNO) for offering some of her resources to a Mobile Virtual Network Operators (MVNO) instead of using the resources for her own End-Users (EUs). We considered a continuum of undecided EUs with two different models:  (i) when EUs need to choose either the MNO or the MVNO, and (ii) when EUs have an outside option. In each of these cases, we considered a non-cooperative framework of sequential games and a cooperative framework of  bargaining game. We characterized the SPNE and the NBS of the sequential and bargaining games, respectively, and provided insights on the interplay of competition and cooperation between the MNO and the MVNO.
\bibliographystyle{IEEEtran}
\bibliography{bmc_article}

\appendices 
\section{Proof of Theorem~\ref{theorem:NoCorner}}\label{section:appendix:corner}

In this section, we prove that there is no corner equilibrium for Stage 2 of the game. Thus, we focus on equilibrium strategies by which $n_L=0$ or $n_F=0$. First,   we prove the case for  $n_L=0$:

\begin{theorem}\label{theorem:n_L=0}
There is no NE strategy for Stage 2 of the game by which $n_L=0$.
\end{theorem}

\begin{proof}
In this case, $x_n\leq 0$. Thus, using \eqref{equ:payoffF}, \eqref{equ:payoffL}, \eqref{equ:EUs_stage4}, the payoffs of SPs would be:
\begin{equation}
\begin{aligned}
\pi_F&=p_F-c-sI^k_F\\
\pi_L&=sI^k_F-\gamma I^k_L
\end{aligned}
\end{equation}

Note that the equilibrium strategies $p^*_L$ and $p^*_F$ should be such that $x_n\leq 0$. Thus, by \eqref{equ:xn}, $p^*_F\leq p^*_L-\frac{I_L-I_F}{I_L}$. Therefore, since the payoff of SP$_F$ is a strictly increasing function of $p_F$, $p^*_F=p^*_L-\frac{I_L-I_F}{I_L}$. We first argue that $p^*_L\leq c$. Then, we prove that there exists no equilibrium strategy in this case. 

We start by proving that  $p^*_L \leq c$. Suppose not and $p^*_L>c$. Then, a deviation by SP$_L$ such that $p'_L=p^*_L-\epsilon$ for $\epsilon>0$ infinitesimal small yields that $0<x_n<1$. In this case, $n_L>0$, and $p'_L>c$. Thus, by \eqref{equ:payoffL} and since $I_F$ and $I_L$ are the same as before, the payoff would be higher after the deviation. Thus, $p^*_L>c$ cannot be an equilibrium strategy. Therefore, $p^*_L\leq c$.  

Now, consider the strategy of SP$_F$, i.e. $p^*_F=p^*_L-\frac{I_L-I_F}{I_L}$. Note that $p^*_F-c\leq 0$. Consider the deviation by SP$_F$ such that $p'_F=p^*_F+\epsilon$ for $\epsilon>0$ infinitesimal small. In this case, $0<x_n<1$. Using \eqref{equ:payoffF} and since  $p^*_F-c\leq 0$ and $I_F$ and $I_L$ are the same as before, the payoff would be higher after the deviation. Thus, $p^*_F$ cannot be an equilibrium strategy. The theorem follows.
\end{proof}

Now, we prove the case for $n_F=0$:

\begin{theorem}\label{theorem:n_F=0}
There is no NE strategy for Stage 2 of the game by which $n_F=0$.
\end{theorem}

\begin{proof}
In this case, $x_n\geq 1$. Thus, using \eqref{equ:payoffF}, \eqref{equ:payoffL}, \eqref{equ:EUs_stage4}, the payoffs of SPs would be:
\begin{equation}
\begin{aligned}
\pi_F&=-sI^k_F\\
\pi_L&=p_L-c+sI^k_F-\gamma I^k_L
\end{aligned}
\end{equation}

Note that the equilibrium strategies $p^*_L$ and $p^*_F$ should be such that $x_n\geq 1$. Thus, by \eqref{equ:xn}, $p^*_L\leq p^*_F-\frac{I_F}{I_L}$. Therefore, since the payoff of SP$_L$ is a strictly increasing function of $p_L$, $p^*_L= p^*_F-\frac{I_F}{I_L}$. We first argue that $p^*_F\leq c$. Then, we prove that there exists no equilibrium strategy in this case. 

We start by proving that  $p^*_F \leq c$. Suppose not and $p^*_F>c$. Then, a deviation by SP$_F$ such that $p'_F=p^*_F-\epsilon$ for $\epsilon>0$ infinitesimal small yields that $0<x_n<1$. In this case, $n_F>0$, and $p'_L>c$. Thus, by \eqref{equ:payoffF} and since $I_F$ and $I_L$ are the same as before, the payoff would be higher after the deviation. Thus, $p^*_F>c$ cannot be an equilibrium strategy. Therefore, $p^*_F\leq c$.  

Now, consider the strategy of SP$_L$, i.e. $p^*_L=p^*_F-\frac{I_F}{I_L}$. Note that $p^*_L-c\leq 0$. Consider the deviation by SP$_L$ such that $p'_L=p^*_L+\epsilon$ for $\epsilon>0$ infinitesimal small. In this case, $0<x_n<1$. Using \eqref{equ:payoffL} and since  $p^*_L-c\leq 0$ and $I_F$ and $I_L$ are the same as before, the payoff would be higher after the deviation. Thus, $p^*_L$ cannot be an equilibrium strategy. The theorem follows.
\end{proof}

Theorems~\ref{theorem:n_L=0} and \ref{theorem:n_F=0} yields the result in Theorem~\ref{theorem:NoCorner}.

\end{document}